\documentclass[journal,onecolumn]{IEEEtran}
\usepackage{setspace}
\usepackage{amsmath,amssymb,epsfig,psfrag,cite,subfigure}
\include{macros}
\usepackage{graphicx}
\usepackage{bm,bbm}

\usepackage[linesnumbered,ruled,vlined]{algorithm2e}
\usepackage{subfigure}
\usepackage{color}

\usepackage{tikz}
\usepackage{pict2e}
\usepackage{amsmath,amsfonts,amsthm} 
\usepackage{commath}
\usepackage{pgfplots}
\usepackage{bbm}
\usepackage{marginnote}

\newtheorem{theorem}{Theorem}
\newtheorem{lemma}{Lemma}

\newtheorem{corollary}{Corollary}
\newtheorem{definition}{Definition}

\newtheorem{assumption}{Assumption}

\theoremstyle{remark}
\newtheorem{remark}{Remark}

\newcommand*\oline[1]{%
	\kern0.1em            
	\vbox{%
		\hrule height 0.4pt 
		\kern0.15ex          
		\hbox{%
			\kern-0.1em       
			$#1$
			\kern-0.1em       
		}
	}
	\kern0.1em            
}

\newcommand*\uline[1]{
	\kern0.1em            
	\vbox{
		\hbox{%
			\kern-0.1em       
			$#1$
			\kern-0.1em       
		}
		\kern0.15ex
		\hrule height 0.4pt 
	}
	\kern0.1em            
}

\def \cR {\mathcal{R}}
\def \cA {\mathcal{A}}

\def \cB {\mathcal{B}}

\def \cH {\mathcal{H}}
\def \cF {\mathcal{F}}

\def \cI {\mathcal{I}}

\def \cU {\mathcal{U}}
\def \cN {\mathcal{N}}
\def \cL {\mathcal{L}}

\def \indic {\mathbbm{1}}

\def \oR {\oline{R}}
\def \uR {\uline{R}}

\DeclareMathOperator{\Var}{\mathrm{Var}}

\renewcommand{\qedsymbol}{\IEEEQED}

\begin{document}

	\title{A Fundamental Limit of Distributed Hypothesis Testing Under Memoryless Quantization}
	\author{\IEEEauthorblockN{Yunus \.Inan, \IEEEmembership{Student Member, IEEE}, Mert Kayaalp, \IEEEmembership{Student Member, IEEE},\\ Ali H. Sayed, \IEEEmembership{Fellow, IEEE}, Emre Telatar, \IEEEmembership{Fellow, IEEE}}\thanks{The authors are with  \'{E}cole Polytechnique F\'{e}d\'{e}rale de Lausanne (EPFL), 1015 Lausanne, Switzerland. Emails: \{yunus.inan, mert.kayaalp, ali.sayed, emre.telatar\}@epfl.ch.\newline A short version of this work is presented at IEEE ICC 2022 \cite{ConfVersionInan}.}}
	\maketitle 
	
	\begin{abstract}
		 We study a distributed hypothesis testing setup where peripheral nodes send quantized data to the fusion center in a memoryless fashion. The \emph{expected} number of bits sent by each node under the null hypothesis is kept limited. We characterize the optimal decay rate of the mis-detection (type-II error) probability provided that false alarms (type-I error) are rare, and study the tradeoff between the communication rate and maximal type-II error decay rate. We resort to rate-distortion methods to provide upper bounds to the tradeoff curve and show that at high rates lattice quantization achieves near-optimal performance. We also characterize the tradeoff for the case where nodes are allowed to record and quantize a fixed number of samples. Moreover, under sum-rate constraints, we show that an upper bound to the tradeoff curve is obtained with a water-filling solution.
	\end{abstract}
	
	
	\section{Introduction}
	A present-day engineering challenge is to cope with amounts of data whose processing requirements exceed the processing capabilities of a single device. Therefore, many state-of-the-art designs share the workload (e.g., data acquisition) among peripheral nodes (e.g., sensors). One particular challenge is the distributed hypothesis testing problem. In this paradigm, peripheral nodes acquire partially informative observations and transmit their information to a fusion center. The center aims to infer the state-of-nature by aggregating information from different nodes. If there were no communication constraints, the fusion center would be able to observe the raw data and perform an optimal test. However, in practice, there may exist communication constraints due to channel imperfections or processing capabilities of the nodes.
	
	To illustrate this situation, consider a vehicle equipped with a collision avoidance system that relies on vehicular communication. In such a communication scheme, information can be received from other vehicles (vehicle-to-vehicle, V2V) or from other objects such as mobile phones, base stations etc. (vehicle-to-everything, V2X).  The collision avoidance system is activated upon detection of a possible collision --- this risky state might be associated with the alternative hypothesis in a binary hypothesis testing setting. V2V and V2X communication protocols, e.g., IEEE 802.11p \cite{IeeeStd}, limit the data rate. Hence, the devices in proximity are required to compress or quantize the data they possess before sending it to the vehicle.

%
	
	The example above can be cast into a canonical distributed hypothesis testing problem under communication constraints. Among many possible ways of restricting communication, we choose to limit the \emph{average} number of bits sent under riskless or ordinary state, which associates with the null hypothesis; and we seek the fundamental limits of a distributed hypothesis testing problem under such assumption. We focus on the case where nodes compress their data with practically-appealing memoryless quantization procedures. More precisely, under such setting, we initially focus on the single-node case and when the average number of bits sent is at most $R$ under the null hypothesis:
	\begin{itemize}
		\item we characterize the optimal decay rate of the type-II error probability under vanishing type-I error probability, given by $\theta^*(R)$, in Theorem \ref{thm:theta} of Section \ref{sec:best_performance_a};
		\item we obtain an upper bound to $\theta^*(R)$ via rate-distortion methods and consequently characterize an unachievable region in Corollary \ref{cor:upper_bnd} of Section \ref{sec:upper_bnd};
		\item we show that with simple lattice-quantization, the upper bound can be approached within $\frac 1 2 \log_2(\frac{\pi e}{2}) \approx 1.047$ bits in Theorem \ref{thm:lattice} of Section \ref{sec:high_rate}; 
		\item we provide the upper bound $\theta_k(R)$ for the $k$-dimensional vector quantization case in Section \ref{sec:multiple_quant}.
	\end{itemize}
	The results for the single-node case are then extended to multiple nodes in Section \ref{sec:multiple_node}, where the problem is formulated under individual communication constraints at nodes, together with a sum-rate constrained formulation.
	
	\section{Related Work}
	The present work contains flavors from both information theoretic and signal processing approaches. We review the related work under these two approaches respectively.
	\subsection{Information Theoretic Approaches}
	Distributed hypothesis testing under communication constraints is a long-standing problem studied by the information theory community. An early work by Ahlswede and Csiszár \cite{Ahlswede} underlies most of the subsequent developments. It is therefore instructive to review their problem setup for a better understanding of the subsequent work. Their setup --- henceforth referred to as the Ahlswede--Csiszár setup --- is as follows. A remote node possesses a sequence $X^n$, while the decision maker possesses a $Y^n$. The pair $(X^n,Y^n)$ is independent and identically distributed (i.i.d.) with distribution $P$ under the null hypothesis ($\cH_0$) and with distribution $Q$ under the alternative hypothesis ($\cH_1$). The decision maker estimates the true hypothesis by using both $Y^n$ and an $nR$-bit side information conveyed by the remote node. The communication constraint is ``hard'' in the sense that $X^n$ is represented with exactly $nR$ bits under both hypotheses. Their aim is to find the fastest exponential decay rate of the type-II error given a prescribed type-I error probability, say $0 < \epsilon < 1$. It turns out that the fastest decay rate does not depend on $\epsilon$, and it is fully characterized for the special case of dependence testing, i.e., when $Q_{XY} = P_XP_Y$ where $P_X$, $P_Y$ are the marginals of $X$ and $Y$ under $\cH_0$. The characterization of the optimal decay rate for the general case turns out to be more involved and it is still unknown although some upper and lower bounds exist.
	
	The Ahlswede--Csiszár setup motivated various subsequent works on distributed hypothesis testing. For instance, \cite{Han_1} presents tighter lower bounds on the optimal decay rate for the Ahlswede--Csiszár setup and further extends the formulation to include zero-rate compression (see also \cite{shalaby}), as well as to include the compression of $Y^n$. The lower bound on the best possible decay rate for the general case is improved in \cite{shimokawa1994}. One may refer to \cite{Han_2} for a comprehensive survey on the literature considering Ahlswede--Csiszár setup and its variants. Subsequent works on communication-constrained hypothesis testing include studies on tradeoffs between type-I and type-II error exponents \cite{Han_exp,weinberger2019}, performance under finite-blocklength regime \cite{watanabe2017}, and under noisy communication \cite{tuncel,sreekumar2020,broadcast_dependence,broadcast}. Further extensions of this problem include interactive protocols \cite{katz,interactive,escamilla}, privacy constraints \cite{mhanna,gunduz_privacy,tan_privacy,liao2018,abbasalipour2022privacy,liao2017privacy}, the additional task of data reconstruction at the receiver \cite{debbah}. For dependence testing, \cite{rahman} concludes that binning schemes are optimal; whereas the recent work \cite{watanabe2022} shows that the performance can be improved with sequential methods for the general case.
	
	The works cited above elaborate on the ``hard'' communication constraints, as Ahlswede and Csiszár did. A recent strand of works relax the ``hard'' communication constraints and study the dependence testing problem by limiting the \emph{expected} number of bits sent. A partial list of the studies adopting this perspective is \cite{Wigger_1,Wigger_2,hamad:hal-03349810,hamad2021,hamad:hal-03349654,benefits}. The current study is also in line with this perspective. We remark that for the special case of dependence testing, since the $X$ marginals are the same under both hypotheses, the expected number of bits conveyed does not depend on the true hypothesis. Our work, on the other hand, focuses on the general case. Hence, given a strategy, the expected number of bits sent might differ under the null and alternative hypotheses; introducing an inherent asymmetry to the problem. We choose to limit the \emph{expected rate under the null hypothesis $\cH_0$}. This choice aligns with the view that $\cH_1$ is a rare high-risk event and necessary communication must take place to detect this event with high probability. A more detailed discussion on such choice is given in Section \ref{sec:problem_formulation}.
	
	Information theoretic approaches could be criticized because they use high-dimensional vector quantization, i.e., the entire block $X^n$ should be observed before being represented with $nR$ bits. A system designed as such may not be desirable for timing-crucial applications, as the decision maker is kept oblivious of the side information until time $n$. Furthermore, for large $n$, such a system is not memory-efficient as the remote node records the whole past and it might also be computationally expensive to compress $X^n$. These observations suggest that low-dimensional quantization could be of interest for low-latency and memory-efficient applications. Such quantization procedures for distributed detection are often studied in the signal processing literature.
	
	\subsection{Signal Processing Approaches}
	
	As mentioned, signal processing approaches are usually centered around low-dimensional quantizer designs. The scalar quantization procedures specialized for the task of binary hypothesis testing aim to keep the dissimilarity between the distributions of the quantizer output under $\cH_0$ and $\cH_1$ as large as possible while representing the output only with $R$ bits. Various methods for evaluating the dissimilarity include calculation of the Kullback--Leibler divergence $D(\cdot||\cdot)$ --- the optimal type-II error rate under vanishing type-I error \cite{Chernoff} or vice versa --- or one may consider the more general Ali--Silvey distances \cite{alisilvey} (or equivalently $f$-divergences \cite{csiszar_shields}) which prove useful for a variety of signal detection problems \cite{kailath}. Notable early studies on quantization for binary hypothesis testing include \cite{kassam,poor_thomas,Gray}. 
	
	Finding the optimal quantizer is in general a daunting task and there is no standard machinery to obtain such quantizers. However, there exists iterative methods
	to find suboptimal quantizers as in \cite{Gray}, or studies on the high-rate quantization regime \cite{poor_highrate,Hero,villard2011}. Some extremal properties of likelihood-ratio quantizers is given in \cite{tsitsiklis_extremal}. Quantizer designs based on privacy and secrecy constraints are studied in \cite{varshney_privacy,mhanna2016}. Error resilient designs are studied in \cite{savas,chen,viswanathan}, as well as Byzantine resilient designs in \cite{marano2009}. A recent work on multilevel quantization is \cite{gul2021}.
	
	A similar trend to that in the information theoretic studies is also observed in the signal processing literature --- the works cited above rely on ``hard'' communication constraints. Different from the existing signal processing literature, we study the fundamental limits under memoryless (scalar) quantization with \emph{expected} rate constraints under $\cH_0$, and provide impossibility results for the subject case. Namely, if the \emph{expected} rate under $\cH_0$ is limited to $R$ bits, then the type-II error rate cannot be greater than $\theta^*(R)$ --- defined in Theorem \ref{thm:theta} --- under vanishing type-I error probability.\\

	\noindent \textbf{Notation:} Random variables are denoted with uppercase letters whereas their realizations are written lowercase, e.g., $X_n$ and $x_n$. $\cB{(\mathbb{R})}$ denotes the Borel algebra of $\mathbb{R}$. For probability measures $P$ and $Q$, $D(P||Q)$ denotes the Kullback--Leibler (KL) divergence and $E_P[\cdot]$, $H_P(\cdot)$, $I_P(\cdot\,;\cdot)$ denote the expectation, entropy, and mutual information under $P$ respectively. All logarithms are taken with natural base unless explicitly stated.

	\section{Problem Formulation}\label{sec:problem_formulation}
	In this work, we study a canonical distributed hypothesis testing setup where $m$ peripheral nodes communicate with a fusion center (Figure \ref{fig:network}). At each time instant $t$, the node $i$ observes data arising from distribution $P^{(i)}$ under the null hypothesis $\cH_0$, and from distribution $Q^{(i)}$ under the alternative hypothesis $\cH_1$. We assume that for all $i$, $P^{(i)}$ is absolutely continuous with respect to $Q^{(i)}$. That is, if $Q^{(i)}(B) = 0$, then $P^{(i)}(B) = 0$ for any $B \in \cB(\mathbb{R})$. The data is independent across nodes, and across time under both hypotheses. Moreover, the data is identically distributed across time. Therefore, the joint distribution of the network until time $t$ and under $\cH_0$ can be characterized on rectangles in $\mathbb{R}^{tm}$ as follows:
	\begin{equation}\label{eq:joint1}
	P(B^t) = \prod_{\tau = 1}^t\prod_{i = 1}^m P^{(i)}( [a_\tau^{(i)}, b_\tau^{(i)}])
	\end{equation}
	where $B_{\tau} := [a_\tau^{(1)}, b_\tau^{(1)}] \times\dots \times [a_\tau^{(m)}, b_\tau^{(m)}]$ and $B^t := B_1\times \dots \times B_t$ are rectangles in $\mathbb{R}^{m}$ and $\mathbb{R}^{tm}$ respectively. By a standard extension theorem, \cite[Theorem 1.7]{Martingales}, $P$ can be extended uniquely to $\cB(\mathbb{R}^{tm})$. Under $\cH_1$, since the independence assumptions are the same, the joint distribution of the network is given exactly by \eqref{eq:joint1}, with $P$'s replaced by $Q$'s.
	
	\begin{figure}[h]
		\centering
		\includegraphics[scale=0.8]{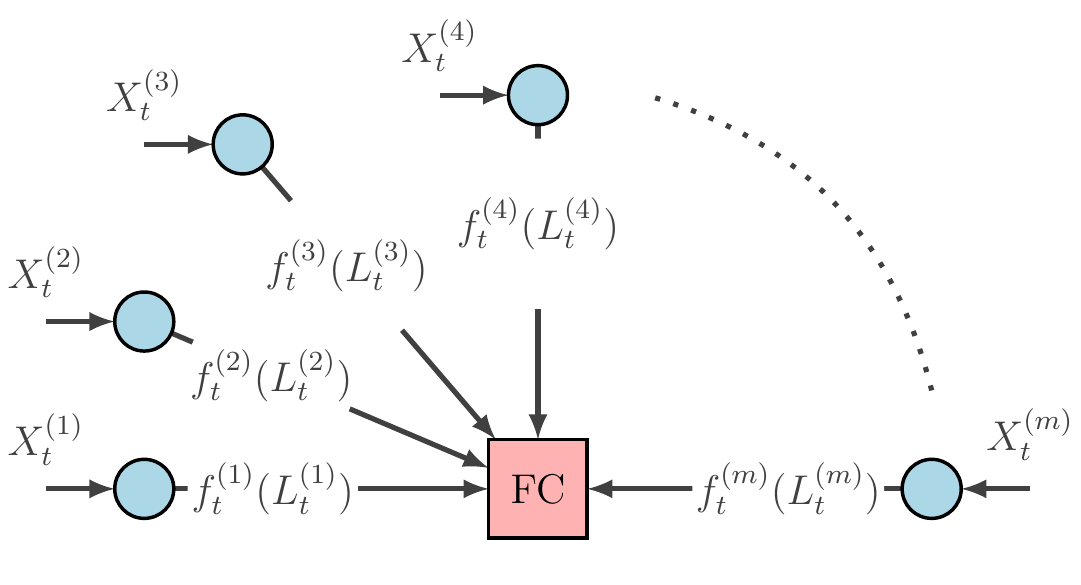}
		\caption{A representation of the setup studied in this work. The peripheral nodes are drawn as blue circles, and the fusion center (FC) is drawn as the red square. At each time instant $t$, node $i$ sends its compressed score $f_t^{(i)}(L_t^{(i)})$, which is solely based on the fresh observation $X_t^{(i)}$.}\label{fig:network}
	\end{figure}
	
	A key assumption in our setup is that \emph{each node $i$ is only aware of $P^{(i)}$ and $Q^{(i)}$, and the fusion center does not have any knowledge about the statistics of the data observed at the nodes}. Such assumption distinguishes our work from many information-theoretic approaches. For instance, in the Ahlswede--Csiszár setup, both the remote observer and the decision maker are aware of the joint distribution. By contrast, our oblivious fusion center trusts the nodes blindly and sums the ``scores" sent by them. Knowing this behavior of the center, nodes prepare their scores accordingly. An example of a score might be the log-likelihood ratio (LLR) of the data observed at time $t$, i.e., node $i$ calculates the LLR $L_t^{(i)}$ based on its freshly observed data $X_t^{(i)}$ as
	\begin{equation}
	L_t^{(i)} := \log\frac{dP}{dQ}(X_t^{(i)}),
	\end{equation}
	sets the score $S_t^{(i)} = L_t^{(i)}$, and passes it through the communication link. Note that the above LLR is well-defined as a Radon--Nikodym derivative due to the absolute continuity of $P^{(i)}$ with respect to $Q^{(i)}$. Suppose each node behaves similarly, i.e., calculates and sends its LLR. Since the data is independent across nodes and across time, under such a strategy, the fusion center receives the sufficient statistic $\sum_{\tau = 1}^t\sum_{i = 1}^m L_\tau^{(i)}$ and is able to perform an optimal test, i.e., a Neyman--Pearson test. However, $P^{(i)}$ and $Q^{(i)}$ can be continuous in general and it is impossible to (i) calculate the LLR with an arbitrarily high precision and (ii) represent the score losslessly with a finite number of bits. Due to these restrictions, the nodes are required to compress (quantize) the data they receive, and send their scores with a finite number of bits at each time instant. The finite-bit score sent by node $i$ at time $t$ is represented by $S_t^{(i)}$ and the fusion center performs a threshold test based on the average score
	\begin{equation}
	\bar S_t := \frac 1 t \sum_{\tau = 1}^t\sum_{i = 1}^m S_\tau^{(i)},
	\end{equation}
	with the estimate being the result of the following test:
	\begin{equation}\label{eq:threshold_test}
	\hat \cH = 
	\begin{cases}
	\cH_0,& \bar S_t \geq \eta_t\\
	\cH_1,& \text{else}
	\end{cases}
	\end{equation}
	where $\eta_t$ is a threshold that can depend on $t$. In addition to the finite-bit constraint, the rate of communication between the nodes and the center may be subject to limitations. As mentioned in the previous section, in this work, we study the distributed hypothesis testing problem under the following communication constraint: \emph{The average number of bits sent under $\cH_0$ must be kept limited.}\\
	
	\begin{remark}\label{rem1} The communication constraint is not symmetric, i.e., there is no constraint under $\cH_1$. This aligns well with many real-world scenarios when $\cH_1$ represents a high-risk situation in which the system is allowed to violate communication constraints in order to identify the risk --- responding to an emergency takes priority over communication constraints ---  recall the collision avoidance example at the beginning of this manuscript. This view of $\cH_1$ also implies that the type-II error must be very rare. In fact, in many hypothesis testing problems, it is desired that the type-II error decays exponentially. This is the approach we follow for the rest of this work.\hfill\qedsymbol
	\end{remark}
	
	\subsection{Memoryless Quantization and the Communication Constraint}
	
	For simplicity, we focus on a single node $i$ at the moment and omit the symbol $(i)$ from the superscripts. In this section, we formally define the memoryless quantization procedures that map the LLR $L_t$ to the score $S_t$, and the communication constraints for such quantization procedures. We first give the definition of a simple function.
	\begin{definition}[Simple function, \cite{Rudin}]\label{def:simple} A function on $\mathbb{R}$ that takes finitely many values is called a simple function. More precisely, let $\alpha_1,\dots,\alpha_n$ be the distinct values of a simple function $f$, then any such $f$ is represented as 
		\begin{equation}\label{eq:simple}
		f(l) = \sum_{k=1}^n\alpha_k\indic\{l \in B_k\}
		\end{equation}
		where $B_1,\dots,B_k \in \cB({\mathbb{R}})$ form a partition of $\mathbb{R}$.\hfill\qedsymbol
	\end{definition}
	We let $S_t = f_t(L_t)$ with a simple function $f_t$. Observe that such procedures are \emph{memoryless} --- quantization at time $t$ depends only on the data arriving at time $t$, and does not depend on past. From Definition \ref{def:simple}, it is clear that $S_t$'s are discrete random variables. For example, if $f_t$ is set as in \eqref{eq:simple}, then for $1 \leq k \leq n$:
	\begin{equation}
	P(S_t = \alpha_k) = P(L_t \in B_k)
	\end{equation}
	and the discrete entropy of $S_t$ under $\cH_0$ is defined as
	\begin{equation}
	H_P(S_t) := -\sum_{k=1}^n P(S_t = \alpha_k)\log P(S_t = \alpha_k)
	\end{equation}
	with $0\log 0: = 0$.
	A discrete random variable can be compressed \emph{losslessly} with a binary code whose expected length is $\ell$, which is bounded as \cite{Wyner,Alon}
	\begin{equation}\label{eq:lossless}
	H_P(S_t)\log_2e - \log_2(H_P(S_t)\log_2e + 1)
	-\log_2 e\leq \ell \leq H_P(S_t)\log_2e.
	\end{equation}
	Therefore, the peripheral node can compress its LLR $L_t$ with a simple function $f_t$, and can represent its score $S_t = f_t(L_t)$ with an average number of bits less than $H_P(S_t)\log_2e$ under $\cH_0$. If we impose
	\begin{equation}\label{eq:rate_constraint}
	\frac 1 t \sum_{\tau = 1}^t H_P(S_\tau) \leq R/\log_2e,
	\end{equation}
	all scores until time $t$ can be represented with an expected number of bits less than $Rt$ under $\cH_0$; and the average number of bits sent over the communication link is kept limited to at most $R$ bits. Constraints formed as in \eqref{eq:rate_constraint} are then suitable candidates for being the communication constraint in our distributed hypothesis testing setting.\\
	
	\begin{remark} The memoryless quantization procedures we consider are practically appealing since the peripheral devices can be designed in a memory-efficient manner. Moreover, the assumption that each node only knows their own $P$'s and $Q$'s allows independent design of the peripheral nodes, as opposed to the joint design of all sensors which may be impractical. Note that without independence across the nodes, joint design might be necessary. We assume that the network subject to this study is designed such that the peripheral nodes have a spatial configuration that yields, or at least approximates, independence across nodes.\hfill\qedsymbol
	\end{remark}
	\subsection{Performance Criteria under Memoryless Quantization}\label{sec:performance}
	As mentioned earlier, the fusion center decides over the hypotheses based on the threshold test given in \eqref{eq:threshold_test}. Under this test, the type-I and type-II error probabilities are defined respectively as
	\begin{equation}\label{eq:errors_def}
	\begin{split}
	\alpha_t &:= P(\bar S_t < \eta_t)\\
	\beta_t &:= Q(\bar S_t \geq \eta_t)
	\end{split}.
	\end{equation}
	For an $\epsilon>0$, we assume that the fusion center sets the threshold to
	\begin{equation}\label{eqn:threshold}
	\eta_t =\frac 1 t \sum_{\tau = 1}^t E_P[S_t] - \epsilon.
	\end{equation} 
	Recalling that the fusion center is unaware of the statistics at the nodes, one might argue that this choice of $\eta_t$ is not valid. However, such adjustment is without loss of generality: If the nodes send the centered version of the scores, i.e., $S_t-E_P[S_t]$, and if the fusion center performs the test based on $\eta_t = -\epsilon$, the performance of the scheme will be equivalent to the scheme where the threshold is chosen as in \eqref{eqn:threshold}. Note that setting $\eta_t = -\epsilon$ does not require any knowledge on $P$'s and $Q$'s. Moreover, the discrete entropy $H_P(S_t)$ does not change under any shift and the communication constraints are not violated. It turns out that the choice in \eqref{eqn:threshold} achieves the optimal curve $\theta^*(R)$ --- to be defined in Theorem \ref{thm:theta} --- and we keep this choice for the rest of the work.
	
	\section{Best Performance Under Memoryless Quantization}
	
	\subsection{Boundary of the Achievable Region}\label{sec:best_performance_a}
	In view of Remark \ref{rem1}, our aim is to drive the type-II error probability to zero as fast as possible while ensuring the type-I error probability  vanishes. In particular, the type-II error probability must decay exponentially. A suitable definition of an achievable region in line with this perspective is given as follows.
	\begin{definition}\label{def:achievable_single}
		Given $P$ and $Q$, $(R,\theta)$ is an achievable pair if there exists a sequence $\{f_t\}$ of simple functions and thresholds $\{\eta_t\}$ such that 
		\begin{itemize}
			\item[(a)] $\frac 1 t \sum_{\tau = 1}^t H_P(S_\tau) \leq R$, for all $t$
			\item[(b)] $\lim_{t \to \infty}\alpha_t = 0$
			\item[(c)] $\liminf_{t \to \infty}\frac 1 t \log \frac 1 {\beta_t} \geq \theta$
		\end{itemize}
		where $S_t = f_t(L_t)$ is the quantized LLR, and $\alpha_t$, $\beta_t$ are the type-I and type-II errors defined in \eqref{eq:errors_def}.\hfill\qedsymbol
	\end{definition}
	Note that the communication constraint imposed in Definition \ref{def:achievable_single}(a) is in terms of \emph{nats} for notational simplicity. The achievable region is then defined as the set of the achievable pairs $(R,\theta)$. The theorem below characterizes the boundary of this region in two parts.
	\begin{theorem}\label{thm:theta}
		Let $\theta^*(R) := \sup\{\theta: (R,\theta)\text{  achievable}\}$ and define 
		\begin{align}\label{eq:theta_single_node}
		\theta_t(R) := \sup_{\{f_1,\dots,f_t\} \in \cF_t(R)} \frac 1 t \sum_{\tau = 1}^t \bigg(E_P[S_\tau]-\log E_P[e^{S_\tau-L_\tau}]\bigg)
		\end{align}
		where $\cF_t(R)$ is the set of all simple real-valued functions $f_1,\dots,f_t$ on $(\mathbb{R},\mathcal{B}(\mathbb{R}))$ such that $\frac 1 t \sum_{\tau = 1}^t H_P(S_\tau) \leq R$.
		Then, the following statements hold.\\
		\noindent(i) Let \begin{equation} \label{eqn:concave_env}
		\theta_1(R) = \sup_{f_1 \in \cF_1(R)} E_P[S_1]-\log E_P[e^{S_1-L_1}].
		\end{equation}
		Then $\lim_{t \to \infty} \theta_t(R)$ equals to the upper concave envelope $\breve\theta_1(R)$ of $\theta_1(R)$.\\
		\noindent (ii)
		\vspace*{-1.55\baselineskip}\begin{flalign*}&\phantom{\mbox{(ii)}}\quad 
		\theta^*(R) = \lim_{t \to \infty}\theta_t(R) = \breve\theta_1(R).&&
		\end{flalign*}
	\end{theorem}
	\begin{proof}
		See Appendix \ref{app:theta}.
	\end{proof}
	Theorem \ref{thm:theta} provides the boundary of the achievable region in a variational form that is reminiscent of a single-letter characterization. However, the optimization problem \eqref{eqn:concave_env} has a non-convex domain, which makes $\cF_1(R)$ a non-convex set. We will therefore consider a relaxed version of the optimization problem \eqref{eqn:concave_env} in the next section.
	
	\subsection{An Upper Bound on the Boundary of the Achievable Region}\label{sec:upper_bnd}
	In order to relax the problem \eqref{eqn:concave_env}, we (i) allow randomized quantization, and (ii) modify the communication constraint to $I_P(S_1;L_1) \leq R$, where $I_P(S_1;L_1)$ is the mutual information between $S_1$ and $L_1$ under $\cH_0$. Note that since $H_P(S_1) \geq I_P(S_1;L_1)$, $H_P(S_1) \leq R$ implies $I_P(S_1;L_1)\leq R$, hence the communication constraints indeed become less stringent. Moreover, the randomized quantization procedures can be represented as channels $p_{V|U} : \cU \times\mathbb{R} \to \mathbb{R}_+$ where for each $u$, $p_{V|U}(v,u)$ is a probability mass function on the finite set $\cU \subset \mathbb{R}$. We further relax the problem by taking $\cU = \mathbb{R}$, hence the possible channels become $p_{V|U} : \cB(\mathbb{R}) \times\mathbb{R} \to \mathbb{R}_+$, where for each $u$, $p_{V|U}(v,u)$ is a probability measure on $\mathbb{R}$. Adopting the modifications we have just described, problem \eqref{eqn:concave_env} then becomes
	\begin{equation}\label{eqn:curve}
	\begin{split}
	\theta_U(R) :=  \sup_{p_{V|U}} &\quad E_P[V]- \log E_P[\exp(V-U)]\\
	\text{s.t.}& \quad I_P(U; V)\leq R
	\end{split}
	\end{equation}
	where $U$ has the same distribution as the LLR $L_1$. Observe that as $R$ increases, the optimization domain is enlarged and thus $\theta_U(R)$ cannot decrease; which shows that $\theta_U(R)$ is non-decreasing. Moreover, $\theta_U(R)$ also captures the behavior at the extremes. Intuitively, if $R \to \infty$, then $V$ can be set equal to $U$ and $\theta_U$ becomes
	\begin{equation}\label{eq:max_val}
	E_P[U] = E_P\bigg[\log\frac{dP}{dQ}\bigg] = D(P||Q)
	\end{equation}
	which is known from Stein's lemma \cite{Chernoff} as the optimal type-II error exponent under vanishing type-I error probability. This intuitive argument will be made rigorous in Lemma \ref{lem:delta_limit}. On the other extreme, if $R = 0$, then the best possible choice is to set $V$ equal to a constant $v$ and $\theta_U$ becomes 
	\begin{equation}
	v - \log E_P\bigg[\bigg(\frac{dP}{dQ}\bigg)^{-1}\bigg] - v = 0,
	\end{equation}
	which is consistent with the fact that the center is not able to infer the true hypothesis when there is no communication.
	
	Another useful characterization of $\theta_U$ is given by the following lemma.
	
	\begin{lemma}\label{lem:dist_rate}
		Let
		\begin{equation}\label{eqn:bad_curve}
		\begin{split}
		\tilde\theta_U(R) :=  \sup_{p_{V|U}} &\quad E_P[V]- E_P[\exp(V-U)]+1\\
		\text{s.t.}& \quad I_P(U; V)\leq R.
		\end{split}
		\end{equation}
		Then, $\theta_U(R) = \tilde \theta_U(R)$.
	\end{lemma}
	\begin{proof}See Appendix \ref{app:dist_rate}.
	\end{proof}
	Observe that $-\tilde \theta_U$ is given by
	\begin{equation}
	\begin{split}
	-\tilde\theta_U(R) =  \inf_{p_{V|U}} &\quad -E_P[V]+ E_P[\exp(V-U)]-1\\
	\text{s.t.}& \quad I_P(U; V)\leq R.
	\end{split}
	\end{equation}
	We highlight the equivalence between $-\tilde\theta_U(R)$ and the distortion-rate function with the distortion function $d(u,v) = -v+e^{v-u}-1$. Since it is known that this curve is convex, $-\tilde\theta_U(R)$ is also convex and consequently, $\tilde\theta_U(R)$ is concave. We then make use of the characterization in Lemma \ref{lem:dist_rate} and conclude that $\theta_U(R)$ is concave as well.
	
	We end this section with the following corollary, which states that $\theta_U$ is a concave upper bound to the boundary of the achievable region given by $\theta^*(R)$.
	
	\begin{corollary}\label{cor:upper_bnd}
		$\theta_U(R) \geq \breve\theta_1(R) = \theta^*(R)$.
	\end{corollary}
	\begin{proof} As \eqref{eqn:curve} is a relaxation of \eqref{eqn:concave_env}, we know that $\theta_U(R) \geq \theta_1(R)$. In addition, $\theta_U(R)$ is concave; then it must also dominate the concave envelope $\breve\theta_1(R)$ of $\theta_1(R)$.
	\end{proof}
\vspace{\baselineskip}
	\begin{remark}
		As stated before, we have shown the equivalence of $\theta_U$ and $\tilde\theta_U$ in Lemma \ref{lem:dist_rate}. Although it might be tempting to work with $\tilde\theta_U$, as it is the optimal value of an optimization problem with a linear objective, \eqref{eqn:bad_curve} takes a smaller value than \eqref{eqn:curve} if a generic $p_{V|U}$ is substituted; thus leading to tighter bounds. This is due to the inequality $\log x\leq x-1$. In view of this observation, we work with the formulation \eqref{eqn:curve} in the sequel.\hfill\qedsymbol
	\end{remark}
	
	\subsection{Calculating the Upper Bound $\theta_U$}
	In this section, and for the rest of the work, we assume all the expectations (including the mutual information $I_P(\cdot\,;\cdot)$) are taken under $P$, and we omit it from the subscripts for brevity. Applying Jensen's inequality to the objective function in \eqref{eqn:curve}, we have
	\begin{equation}\label{eq:jensen}
	\begin{split}
	&E[V]- \log E[\exp(V-U)]\\
	&\leq E[V] - E[V-U] = D(P||Q).
	\end{split}
	\end{equation}
	Thus, it is also convenient to study the gap to $D(P||Q)$. The gap function $\delta_U(R) := D(P||Q)-\theta_U(R)$ is then straightforwardly given by
	\begin{equation}\label{eq:gap}
	\begin{split}
	\delta_U(R) = \inf_{p_{V|U}} & \quad  \log E[\exp(V-U)] - E[V-U]\\
	\text{s.t.}& \quad I(U; V)\leq R.
	\end{split}
	\end{equation}
	Note that since $\theta_U$ is concave and non-decreasing, $\delta_U(R)$ is convex and non-increasing by definition, and the following lemma explains the limiting behavior as $R\to \infty$.
	\begin{lemma}\label{lem:delta_limit} $\lim_{R\to\infty}\delta_U(R) = 0$. Consequently, $\lim_{R \to \infty}$ $ \theta_U(R) = D(P||Q)$.
	\end{lemma}
	\begin{proof}
		See Appendix \ref{app:delta_limit}.
	\end{proof}
	We highlight that Lemma \ref{lem:delta_limit} holds even if $P$ and $Q$ do not admit densities. Now, we intend to derive an upper bound for $\delta_U$. Let $Z := V-U$. Then \eqref{eq:gap} is equivalent to
	\begin{equation}\label{eq:delta_Z}
	\begin{split}
	\delta_U(R) = \inf_{p_{Z|U}} & \quad  \log E[\exp(Z)] - E[Z]\\
	\text{s.t.}& \quad I(U; U+Z)\leq R.
	\end{split}
	\end{equation}
	A simple upper bound to $\delta_U(R)$ can be obtained by choosing $Z$ as a Gaussian random variable independent of $U$. With such choice, we have
	\begin{equation}\label{eq:upper_derivation}
	\begin{split}
	I(U; U+Z) &= h(U+Z)-h(U+Z|U)\\
	&= h(U+Z)-h(Z|U)\\
	&= h(U+Z)-h(Z)\\
	&\leq \frac 1 2 \log\bigg(1+\frac{\Var(U)}{\Var(Z)}\bigg)
	\end{split}
	\end{equation}
	where $h(\cdot)$ denotes the differential entropy. Observe that $U+Z$ always admits a probability density; $U$ need not be continuous. However we assume $U$ is square integrable such that $\Var(U)$ exists. Furthermore, for a Gaussian $Z$
	\begin{equation}
	\begin{split}
	\log E[\exp(Z)] - E[Z] &= E[Z] + \log(e^{\frac 1 2 \Var(Z)}) - E[Z]\\
	&=\frac 1 2 \Var(Z).
	\end{split}
	\end{equation}
	Denoting the variance of $Z$ by $\sigma^2$, observe that the parametric curve
	\begin{equation}\label{eq:parametric_upperbnd}
	R = \frac 1 2 \log\bigg(1+\frac{\Var(U)}{\sigma^2}\bigg),\quad \delta = \frac 1 2 \sigma^2
	\end{equation}
	lies above $\delta_U(R)$, and equivalently
	\begin{equation}\label{eq:bad_bound}
	\delta_U(R) \leq \frac{\Var(U)}{e^{2R}-1} =: g_U(R).
	\end{equation}
	The bound \eqref{eq:bad_bound} is however not tight at low rates. Observe that as $R \to 0$, the right-hand side of \eqref{eq:bad_bound} tends to infinity although we know that the gap $\delta$ can at most be $D(P||Q)$ --- see \eqref{eq:jensen}. The bound can be strengthened as follows: Since we know $\delta_U$ is convex with $\delta_U(0) = D(P||Q)$, and $\delta_U(R) \leq g_U(R)$, any line segment connecting $(0,D(P||Q))$ with the curve $g_U(R)$ lies above $\delta_U(R)$. Among such line segments, the one which is tangent to $g_U(R)$ gives the tightest bound.
	
	After obtaining this simple upper bound, we direct our attention to the calculation of $\delta_U(R)$. Note that the objective function in \eqref{eq:delta_Z} is concave. This is because $\log E[e^{Z}]$ is concave, and $E[Z]$ is linear in $p_{Z|U}$. Hence, it is a concave minimization problem, and might a priori require examining all extreme points of the feasible set. However, we now show that the problem can be formulated as a convex minimization, circumventing the combinatorial challenge. First, note that both the objective function and the constraint in \eqref{eq:delta_Z} remain unchanged if we add a constant to $Z$. Thus, centering $Z$ does not change the feasible region in \eqref{eq:delta_Z}. Consequently, we can add the constraint $E[Z] = 0$ to our problem without changing its value, which yields an equivalent formulation of \eqref{eq:delta_Z}:
	\begin{equation}
	\begin{split}
	\delta_U(R) = \inf_{p_{Z|U}} & \quad  \log E[\exp(Z)]\\
	\text{s.t.}& \quad I(U; U+Z)\leq R\\
	& \quad E[Z] = 0
	\end{split}
	\end{equation}
	Any infimizer of the above problem also infimizes the optimization problem with the objective function $E[\exp(Z)]$, and the optimal value of the former problem is the logarithm of the optimal value of the latter. Further note that the objective function becomes linear when changed to $E[\exp(Z)]$. The latter problem is formulated as the convex program
	\begin{equation}\label{eq:dual_formulation}
	\begin{split}
	\Delta_U(R) :=  \inf_{p_{Z|U}} & \quad  E[\exp(Z)]\\
	\text{s.t.}& \quad I(U; U+Z)\leq R\\
	& \quad E[Z] = 0
	\end{split}
	\end{equation}
	with $\log\Delta_U(R) = \delta_U(R)$. Observe that $\Delta_U(R)$ is convex, non-decreasing, and is finite at every $R \geq 0$ --- check the feasible choice $Z = -U  + E[U]$ and observe $\Delta_U(R) \leq E[e^{-U}]e^{E[U]} = e^{D(P||Q)}$. Therefore, $\Delta_U(R)$ can be expressed as
	\begin{equation}
	\Delta_U(R) = \sup_{\lambda> 0} \cL(\lambda) -\lambda R
	\end{equation}
	where
	\begin{equation}
	\begin{split}
	\cL(\lambda) :=  \inf_{p_{Z|U}} & \quad  E[\exp(Z)] + \lambda I(U; U+Z) \\
	\text{s.t.}
	& \quad E[Z] = 0.
	\end{split}
	\end{equation}
	 Each $\lambda > 0$ describes a straight line $\Delta + \lambda R = \cL(\lambda)$ in the $(R,\Delta)$ plane. $\Delta_U(R)$ is the supremum in the $\Delta$ axis of these lines. The generalized inverse of $\Delta_U(R)$, $R_U(\Delta)$, is then the supremum of these lines in the $R$ axis,
	\begin{equation}
	\begin{split}
	R_U(\Delta) &= \sup_{\lambda > 0} \frac 1 {\lambda} \cL(\lambda) - \frac 1 \lambda \Delta\\
	&= \sup_{\eta > 0} \eta \cL\bigg(\frac 1 \eta\bigg) - \eta \Delta
	\end{split}
	\end{equation}
	which is identical to the following convex problem for $\Delta > 1$.
	\begin{equation}\label{eq:ru_delta}
	\begin{split}
	R_U(\Delta) = \inf_{p_{Z|U}} & \quad  I(U; U+Z)\\
	\text{s.t.}& \quad E[\exp(Z)]\leq \Delta\\
	& \quad E[Z] = 0.
	\end{split}
	\end{equation}
	
	An important direction is to obtain a closed-form lower bound for $R_U$, which consequently gives a lower bound for $\delta_U$. Such a lower bound characterizes an unachievable region as $\delta_U$ is a lower bound to the boundary curve of the achievable region. 
	\begin{assumption}
		For the rest of the work, we assume that $U$ admits a probability density $p_U$. Hence the differential entropy $h(U)$ is well-defined (but not necessarily finite). \hfill\qedsymbol
	\end{assumption}
	Note that \eqref{eq:ru_delta} is exactly the same as the rate-distortion formulation except for the additional constraint $E[Z] = 0$. This special structure allows us to derive a lower bound based on maximum-entropy principles, which also led Shannon to derive the well-known lower bound for the rate-distortion problem under mean-square distortion \cite{shannon}. We shall use the same machinery as well. Note that 
	\begin{equation}
	\begin{split}
	I(U;U+Z) &= h(U) - h(U|U+Z)\\
	&= h(U) - h(Z|U+Z)\\
	&\geq h(U) - h(Z)
	\end{split}
	\end{equation}
	where the last inequality is due to the property ``conditioning reduces entropy''. Hence, we obtain
	
	\begin{equation}\label{eq:lower_bnd}
	\begin{split}
	R_U(\Delta) \geq \inf_{p_{Z|U}} & \quad  h(U) -h(Z)\\
	\text{s.t.}& \quad E[\exp(Z)]\leq \Delta\\
	& \quad E[Z] = 0.
	\end{split}
	\end{equation}
	
	Since the new objective function depends only on the marginal of $Z$, the problem above is equivalent to finding a maximum-entropy distribution $p_Z$ that satisfies the constraints $E[e^{Z}]\leq \Delta$ and $E[Z] = 0$. The problem can now be formulated as
	\begin{equation}\label{eq:max_ent}
	\sup_{p_Z} \ h(Z) \quad \text{s.t.} \quad E[e^{Z}] \leq \Delta,\quad  E[Z]= 0.
	\end{equation}
	
	The entropy maximizing distribution can be found with the methods in \cite[Chapter 12]{Cover2006} and is given by
	\begin{equation}
	f(z) = \frac{\beta^{\alpha}}{\Gamma(\alpha)}\exp(\alpha z - \beta e^{z}),\quad  \alpha,\beta > 0.
	\end{equation}
	Observe that $f(z)$ is the distribution of the logarithm of a Gamma random variable, i.e., $Z = \log G$ where $G \sim \text{Gamma}(\alpha,\beta)$. The following entities have closed form expressions:
	\begin{equation}\label{eq:log_gamma}
	\begin{split}
	E[e^{Z}] &= \frac{\alpha}{\beta}\\
	E[Z] &=  \psi(\alpha)-\log\beta\\
	h(Z) &=  \log \Gamma(\alpha) - \alpha \psi(\alpha) + \alpha
	\end{split}
	\end{equation}
	where $\Gamma(.)$ and $\psi(.)$ are gamma and digamma functions respectively. Note that $\log E [e^{Z}] - E[Z] = \log \alpha-\psi(\alpha)$ and does not depend on $\beta$. Substituting \eqref{eq:log_gamma} into \eqref{eq:lower_bnd}, we have just proved
	\begin{lemma}\label{lem:parametric}Define the parametric curve
		\begin{equation}\label{eqn:lower_gap}
		\begin{split}
		\uR_{U}(\alpha) &= h(U) - \log \Gamma(\alpha) + \alpha \psi(\alpha) - \alpha,\\
		\delta(\alpha) &= \log \alpha-\psi(\alpha),\quad \alpha > 0,
		\end{split}
		\end{equation}
		where $\Gamma(.)$, $\psi(.)$ are gamma and digamma functions respectively. Then \underbar{$R$}$_U(\delta) \leq R_U(\delta)$.\hfill\qedsymbol
	\end{lemma}
	
	In comparison, the parametric curve in \eqref{eq:parametric_upperbnd} gives the upper bound 
	\begin{equation}\label{eq:upper_bnd}
	\oR_U(\delta) = \frac 1 2 \log\bigg(1 + \frac{\Var(U)}{2\delta}\bigg).
	\end{equation}
	Both $\oR_U$ and $\uR_U$ are depicted in Figure \ref{fig:comparison} for a Gaussian $U$. As discussed before, the upper bound $\oR_U(\delta)$ is not tight at low rates since we know that $R_U(\delta) = 0$ at $\delta = D(P||Q)$, and the convexity of $R_U(\delta)$ enables tightening the upper bound by drawing the tangent line from the point $(D(P||Q),0)$ to $\oR_U$. This straight line bound is denoted as SL in Figure \ref{fig:comparison}.
	\begin{figure}[h!]
		\centering
		\includegraphics[scale = 0.65]{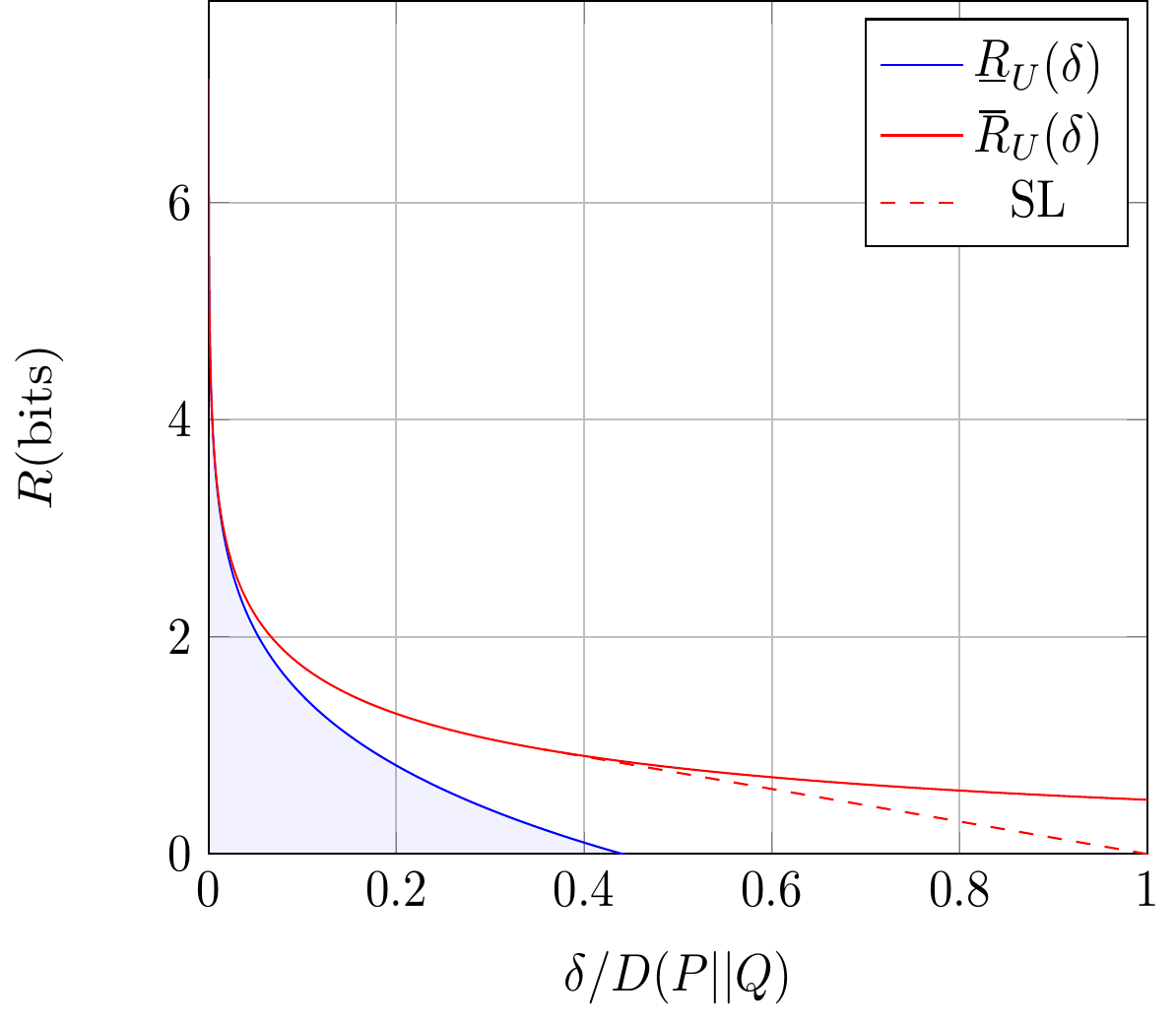}
		\caption{Bounds for $R_U$ curve for the case where $X \sim \cN(0,1)$ under $\cH_0$ and $X \sim \cN(\mu,1)$ under $\cH_1$ for $\mu = \sqrt{20}$. $U$ has the same distribution as the LLR $L \sim \cN(10, 20)$. The lower bound $\uR_U$ is drawn with blue color and the shaded region underneath is unachievable. The upper bound $\oR_U$ is drawn with red color, and its tightened version is drawn with a dashed line, denoted as SL. The true $R_U$ curve lies between $\uR_U$ and SL.}
		\label{fig:comparison}
	\end{figure}
	
	\subsection{Asymptotic Behavior of $R_U(\delta)$}
	Although Figure \ref{fig:comparison} suggests that $\uR_U$ and $\oR_U$ match closely at high rates, it is not evident if they tend to infinity at the same rate. Therefore, the asymptotic behavior of the exact $R_U$ is still unknown.
	We will characterize this behavior in this section. We first derive another upper bound than $\oR_U$ and show that this new upper bound behaves the same as the lower bound $\uR_U$ asymptotically. Once again, refer to \eqref{eq:upper_derivation} and observe for a Gaussian $Z$ with variance $v$ and independent of $U$,
	\begin{equation}\label{eq:fisher}
	\begin{split}
	I(U;U+Z) = h(U+\sqrt{v} \tilde Z)-h(\tilde Z) - \frac 1 2\log{v},
	\end{split}
	\end{equation}
	where $\tilde Z$ is a standard Gaussian random variable. We obtain an upper bound to $h(U+\sqrt{v} \tilde Z)$ with a different method. Suppose $U$ has a differentiable probability density $p_U$. We use De Brujin's identity \cite[Chapter 17]{Cover2006}, which states
	\begin{equation}
	\frac{\partial}{\partial v} h(U + \sqrt{v} \tilde Z) = \frac 1 2 J(U+\sqrt{v} \tilde Z),
	\end{equation}
	where
	\begin{equation}
	J(X) := E\bigg[\Big(\frac{\partial}{\partial x} \log p_X(x)\Big)^2\Big|_{x = X}\bigg]
	\end{equation}
	is the Fisher information of a random variable $X$ with differentiable density $p_X$. We then resort to Taylor's theorem which implies
	\begin{equation}
	h(U+\sqrt{v} \tilde Z) \leq h(U) + \frac{v} 2 \sup_{s\geq 0} J(U + \sqrt{s}\tilde Z).
	\end{equation}
	A well-known convolution inequality for Fisher information states \cite[Chapter 17]{Cover2006} for random variables $X$ and $Y$ with differentiable densities
	\begin{equation}
	\frac 1 {J(X+Y)} \geq \frac 1 {J(X)} + \frac 1 {J(Y)}.
	\end{equation}
	Therefore, 
	\begin{equation}
	\begin{split}
	h(U+\sqrt{v} \tilde Z) &\leq h(U) + \frac{v} 2 \sup_{s\geq 0} \frac{J(U)J(\sqrt{s}\tilde Z)}{J(U)+J(\sqrt{s}\tilde Z)}\\
	& \stackrel{(a)}{=} h(U) + \frac{v} 2 \sup_{s\geq 0} \frac{J(U)}{s J(U)+1}\\
	& \stackrel{(b)}{=} h(U) + \frac{v} 2 J(U)
	\end{split}
	\end{equation}
	where (a) follows from $J(\sqrt{s}\tilde Z) = \frac{1}{s}$, and (b) follows from the fact that $J(U)$ is always non-negative. Substituting this upper bound into \eqref{eq:fisher}, we obtain
	\begin{equation}
	\begin{split}
	I(U;U+Z) &\leq h(U) + \frac v 2 J(U) -h(\tilde Z) - \frac 1 2\log{v}\\
	& = h(U) + \frac v 2 J(U) - \frac 1 2\log(2\pi e v).
	\end{split}
	\end{equation}
	Referring to \eqref{eq:parametric_upperbnd}, we have $v = 2\delta$ and obtain another upper bound to $R_U$ as
	\begin{equation}\label{eq:match_upper_bnd}
	R_U(\delta) \leq h(U) + \delta J(U) - \frac 1 2 \log(4 \pi e \delta).
	\end{equation}
	We intend to obtain a matching lower bound using $\uR_U$. To this end, we use the following inequalities valid for $\alpha > 0$ \cite[5.11(ii)]{NIST:DLMF}:
	\begin{equation}\label{eq:gamma_bounds}
	\begin{split}
	\log\Gamma(\alpha) \leq \alpha\log \alpha - \alpha - \frac 1 2 \log \alpha + \frac 1 2 \log(2\pi) + \frac {1}{12 \alpha}\\
	\log \alpha - \frac {1}{2\alpha} - \frac{1}{12\alpha^2} \leq \psi(\alpha) \leq \log \alpha - \frac {1}{2\alpha}.
	\end{split}
	\end{equation}
	Using \eqref{eq:gamma_bounds} we obtain from \eqref{eqn:lower_gap}
	\begin{equation}
	\begin{split}
	\uR_{U}(\alpha) \geq h(U) - \frac 1 2 \log(2\pi e \alpha^{-1}) - \frac 1 {6\alpha},\quad \delta(\alpha) \geq \frac{1}{2\alpha}.
	\end{split}
	\end{equation}
	Thus,
	\begin{equation}\label{eq:match_lowerbnd}
	R_U(\delta) \geq \uR_{U}(\delta) \geq h(U) - \frac 1 2 \log(4\pi e \delta) - \frac{\delta}{3}.
	\end{equation}
	Comparing \eqref{eq:match_lowerbnd} with \eqref{eq:match_upper_bnd}, one can characterize the high-rate behavior of $R_U$. We conclude this section with the following theorem that gives the asymptotic behavior.
	\begin{theorem}\label{thm:asymp_rate}
		\begin{equation} 
		\lim_{\delta \to 0} R_U(\delta) - h(U) + \frac 1 2 \log(4\pi e \delta) = 0.
		\end{equation}\hfill\qedsymbol
	\end{theorem}
	
	\section{High-Rate Regime and Performance under Vector Quantization}
	
	\subsection{High-Rate Lattice Quantization}\label{sec:high_rate}
	
	Until this point, we have investigated the behavior of $R_U$ and characterized its exact asymptotic behavior. However, we have not yet proposed a concrete quantization scheme that attains $(R,\delta)$ pairs comparable with $R_U$. In this section, we will show that with simple quantization schemes, $R_U$ can be closely approached at high rates. More specifically, we study lattice quantization procedures --- a detailed reference is \cite{Conway}. As we focus on scalar (memoryless) quantization in one dimension, the quantization procedures we consider are simply described as 
	\begin{equation}
	q_r(U) := 2r \arg\min_{k \in \mathbb{Z}} |U -2kr|
	\end{equation}
	where $r$ is the covering radius. Consequently, $V = q_r(U)$ is a quantized version of $U$ with $|V-U| \leq r$. A visual representation is given in Figure \ref{fig:one_dim}.
	
	\begin{figure}
		\centering
		\includegraphics[scale=0.7]{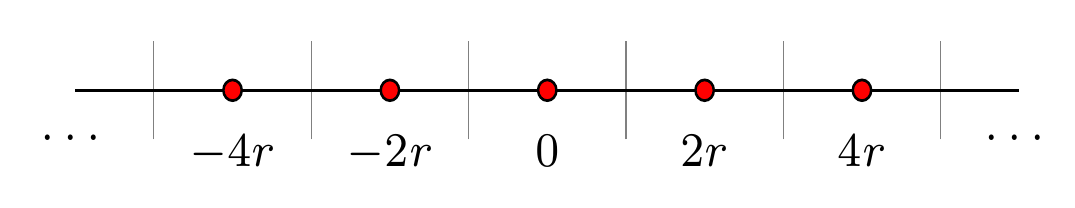}
		\caption{A visualization of $q_r(\cdot)$. The output is set to $2kr$ (drawn as red dots) whenever the input falls into the bin (separated with vertical lines) corresponding to $k$. Each bin is of radius $r$.}\label{fig:one_dim}\vspace*{-\baselineskip}
	\end{figure}

	At this point, we would like to relate the radius $r$ to the gap $\delta$. Referring to \eqref{eq:gap}, under lattice quantization $q_U(r)$, the gap is given by
\begin{equation}
\begin{split}
\delta &= \log E[e^{V-U}] - E[V-U]\\
&\leq E[e^{V-U}] - E[V-U] -1\\
&=  E[e^Z-Z-1]
\end{split}
\end{equation}
and since $|Z| \leq r$ surely, $e^Z-Z-1 \leq e^{r}-r-1$. Consequently, 
\begin{equation}\label{eq:delta_ach}
\delta \leq e^{r}-r-1,
\end{equation}
which suggests that in the small-$r$ regime, $\delta$ behaves quadratically. In fact, if $r \leq D$ for a constant $D$, then
\begin{equation}
\begin{split}
\delta &\leq e^{r}-r-1 = \sum_{k = 2}^\infty \frac{r^k}{k!}\\ &\leq \frac{r^2}{D^2}\sum_{k = 2}^\infty \frac{D^k}{k!} = \frac{r^2}{D^2}(e^{D}-D-1).
\end{split}
\end{equation}
The next step is to relate $r$ with $H(V)$, which is an upper bound to the expected length of an optimal lossless code as discussed in \eqref{eq:lossless}. Under mild regularity conditions on the distribution of $U$, the asymptotic behavior of $H(V)$ when $r \to 0$ can be characterized. 
\begin{theorem}[\kern-3pt\cite{csiszar_quantization}]\label{thm:csiszar}
	Suppose $H(\lfloor U\rfloor) < \infty$. Then, 
	\begin{equation}
	\lim_{r\to 0} H(V) + \log(2r) = h(U).
	\end{equation}
	\hfill\qedsymbol
\end{theorem} Using the above theorem, we have
	\begin{equation}\begin{split}
	H(V) &\leq h(U)-\frac 1 2 \log(4\pi e\delta)+ \frac 1 2 \log(\pi e (e^{D} - D -1)D^{-2})  + f(r)
	\end{split}
	\end{equation}
	where $f(r)$ is a function such that $\lim_{r\to0 } f(r) = 0$. Comparing the above with the asymptotic behavior of $R_U$ given in Theorem \ref{thm:asymp_rate}, we conclude that with high-rate lattice quantization, one can approach the boundary of the achievable region with at most $\frac 1 2 \log_2(\pi e (e^{D} - D -1)D^{-2})$ bits of difference. As $D \to 0$, the difference term tends to $\frac 1 2 \log_2\big(\frac{\pi e}{2}\big) \approx 1.047$ bits. Remembering that we work in the high-rate regime, i.e., we are allowed to send a large number of bits, a $1.047$-bit gap from the optimal curve does not seem to be significant. We summarize our results on the high-rate quantization as follows.
	\begin{theorem}\label{thm:lattice}
		Suppose $H(\lfloor U \rfloor)$ is finite. Then with one-dimensional lattice quantization of sufficiently small radius, the lower bound to the optimal curve $R_U$ can be approached within $\frac 1 2 \log_2(\frac{\pi e}{2}) \approx 1.047$ bits.\hfill\qedsymbol
	\end{theorem}

Although Theorem \ref{thm:lattice} quantifies the gap in the limit $r \to 0$, one may also be interested to find an upper bound on $H(V)$ for strictly positive values of $r$. To this end, one might need more stringent regularity conditions than those of Theorem \ref{thm:csiszar} and work with nicely-behaved distributions. For the moment, consider $\tilde V = V + W$, where $W$ is independent of  $V$ and uniformly distributed in $[-r,r]$. Observe that the probability density of $\tilde V$ is a ``quantized'' version of the probability density of $U$. If $U$ has a nicely-behaved distribution and if $r$ is small, then the distribution of $\tilde V$ will not be very different from that of $U$; which is desirable for the sake of analysis. The family of the aforementioned nicely-behaved distributions are defined as follows.
	
	\begin{definition}[$v$-regular density, \cite{Kostina}]\label{defn:vregular}
		Given $v:\mathbb{R}\to\mathbb{R}$, a continuous and differentiable density function $p$ is called $v$-regular if $\big|\frac{d}{du}p(u)\big| \leq v(u)p(u)$.\hfill\qedsymbol
	\end{definition}
	In \cite[Theorem 8]{Kostina}, it has been proved that if $U$ has a $v$-regular density, then
	\begin{equation}\label{eqn:lattice}
	H(V) \leq h(U) - \log 2r + 2rC_U(r)
	\end{equation}
	where $C_U(r)$ is a function of $r$ depending on the density of $U$ and on the function $v$. Furthermore, if $v$ is Lipschitz-continuous almost everywhere and if $E[v(U)]$ is finite, then $C_U(r)$ can be shown to be bounded for finite $r$ --- see Appendix \ref{app:lipschitz}. In particular, if $v$ has Lipschitz constant $L$,
	\begin{equation}
	C_U(r) \leq \sqrt{J(U)} + 2Lr.
	\end{equation}
	We then obtain the following parametric curve
	\begin{equation}\label{eqn:lattice_upperbnd}
	\begin{split}
	R^{(L)}_U(r) &:= h(U) - \log 2r + 2r\sqrt{J(U)} + 4Lr^2\\
	\delta(r) & := e^{r}-r-1
	\end{split}
	\end{equation}
	which is an upper bound to the $(R,\delta)$ pairs achievable with lattice quantization. Figure \ref{fig:high_rate} illustrates the comparison of the lattice upper bound $R_U^{(L)}$ and the lower bound $\uR_U$ at high rates. The 1.047-bit gap in between is clearly observed.\\
	
	\begin{remark}\label{rem:covering}
		As discussed in \cite{Kostina}, the gap $\frac 1 2 \log_2\big(\frac{\pi e}{2}\big) \approx 1.047$ is due to the covering inefficiency of the one-dimensional lattice. If we perform a similar analysis under the mean-square distortion, the gap turns out to be exactly the same \cite{Kostina}. This is expected as the gap $\delta = \log(E[e^{V-U}]) - E[V-U]$ --- despite not being a distortion function --- behaves like $r^2/2$ for small $r$. Observe that for sufficiently smooth densities, the mean-square error under lattice quantization behaves exactly the same for sufficiently small $r$.\hfill\qedsymbol
	\end{remark}
	
	\begin{figure}[h!]\centering
		\includegraphics[scale = 0.65]{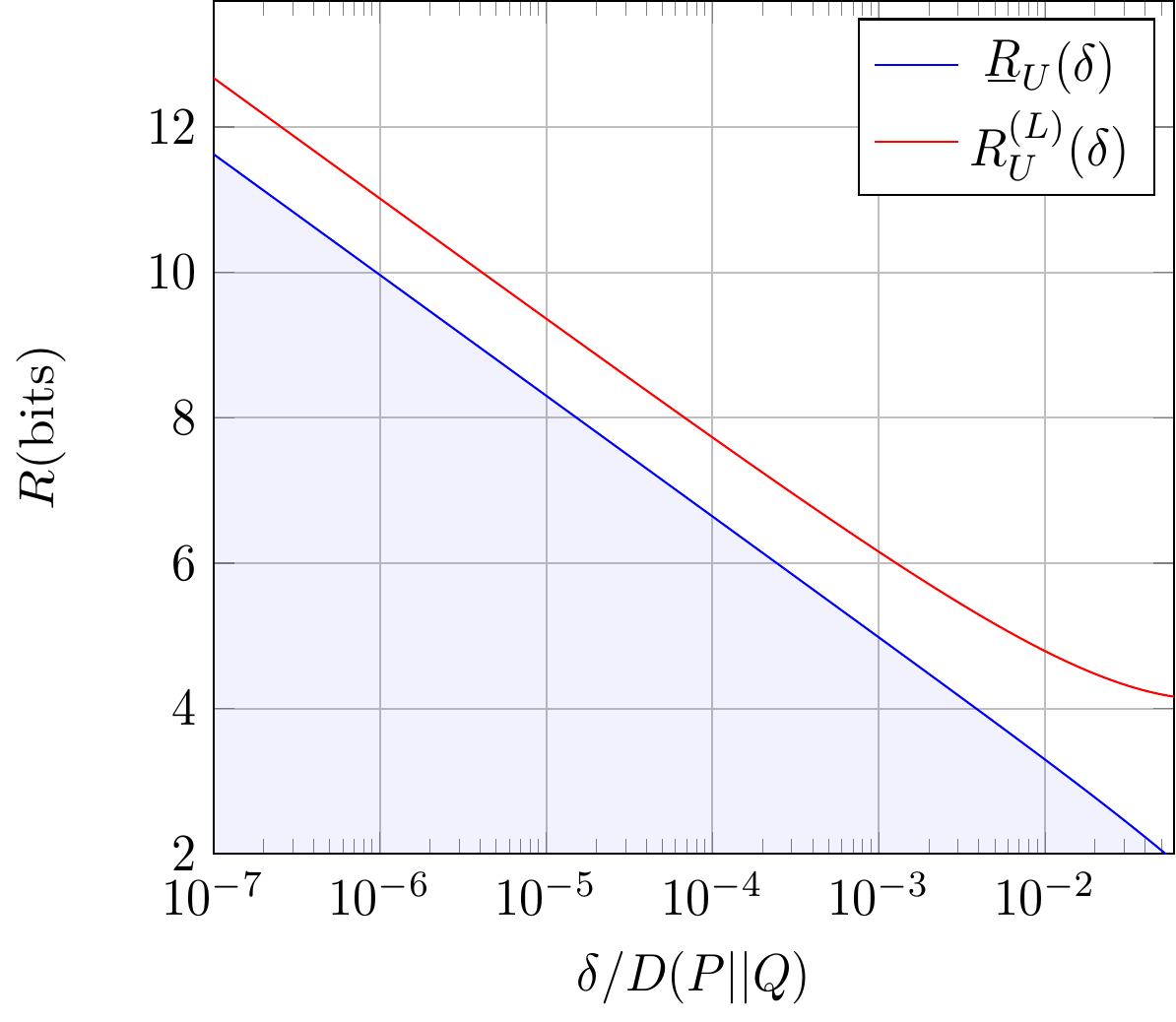}
		\caption{High-rate behaviors of the lower bound $\uR_U$, and the curve $R_U^{(L)}$ achievable with one-dimensional lattice quantization. Same hypothesis testing setup in Figure \ref{fig:comparison} is considered, where $U$ has a Gaussian distribution. It is not difficult to see that a Gaussian distribution is $v$-regular. The 1.047-bit difference between $R_U^{(L)}$ and $\uR_U$, mentioned in Theorem \ref{thm:lattice}, is visible.}
		\label{fig:high_rate}
	\end{figure}
	
	In light of our results in this section, the remote node's strategy in the high-rate regime is apparent. At time $t$, the node (i) calculates its LLR $L_t$, (ii) obtains the lattice-quantized score $S_t = q_r(L_t)$, and (iii) sends $S_t$ with an optimal variable-length lossless code designed for $P$, i.e., for $\cH_0$. This strategy ensures the approach to the optimal curve within $1.047$ bits.
	
	One might ask what is the expected number of bits sent under $\cH_1$ although the code is designed for $\cH_0$. It is known that if the true distribution of the quantized score $S_t$ is given by $Q(S_t)$, then an optimal lossless code designed for $P(S_t)$ yields the expected number of bits at most 
	\begin{equation}\begin{split}
	H_Q(S_t)\log_2e + D(Q(S_t) || P(S_t))\log_2e\leq H_Q(S_t)\log_2e + D(Q||P)\log_2e,
	\end{split}
	\end{equation} where the inequality is due to the data processing inequality. Hence, if $Q$ is also absolutely continuous with respect to $P$, $D(Q||P)$ is finite, and the expected number of bits sent under $\cH_1$ is finite as well.
	
	We end this section by raising the following question: ``Is it possible to eliminate the $\frac 1 2 \log_2\big(\frac{\pi e}{2}\big) $ gap with more efficient lattice coverings?'' For high-dimensions, it is known that covering-efficient lattices exist \cite{Kostina}. Hence, an obvious attempt would be to allow the quantization of multiple samples, i.e., at time $tk$, the remote node records $L_{(t-1)k+1},\dots,L_{tk}$ and sends the $k$-dimensional lattice-quantized version. Although this approach might alleviate the covering inefficiency problem, it is not certain that for such procedures the $R_U$ curve remains the same. We shall study in the next section the behavior of $R_U$ when vector quantization is allowed.

	\subsection{Best Performance under Vector Quantization}\label{sec:multiple_quant}
	
	This section addresses the problem of quantizing multiple samples instead of one. We continue to study memoryless schemes, that is, at time $tk$, the $k$-tuple of LLRs $(L_{(t-1)k+1},\dots,L_{tk})$ is quantized and sent. We first highlight a key observation in the proof of Theorem \ref{thm:theta} --- given in Appendix \ref{app:theta}. Observe that for a choice of quantization function $f$, the (optimal) Neyman--Pearson test pertaining to the quantized $S_t = f(L_t)$ yields the type-II error rate $D(P(S)||Q(S))$ and the rate is optimized over possible $f$'s to obtain $\theta(R)$. Adapting this observation to the vector quantization case, we have the score $S_{t,k} = f(L_{(t-1)k+1},\dots,L_{tk})$, where $f:\mathbb{R}^k \to \mathbb{R}$ is a simple function, and we want to optimize $D(P(S)||Q(S))$ under the constraint $I(L_1,\dots,L_k;S) \leq kR$ to obtain an upper bound. Using the Donsker--Varadhan representation of $D(P(S)||Q(S))$ as we did in the proof of Theorem \ref{thm:theta}, we therefore have the upper bound to the best achievable type-II error exponent, analogous to \eqref{eqn:curve}:
	\begin{equation}\label{eq:multi_suff}
	\begin{split}
	\tilde\theta_{L,k}(R) :=  \frac 1 k \bigg(\sup_{p_{S|L_1,\dots,L_k}} & E_P[S]- \log E_P\Big[\exp\Big(S-L_{1,k}\Big)\Big]\bigg)\\
	\text{s.t.}& \quad I_P(L_1,\dots,L_k; S)\leq kR.
	\end{split}
	\end{equation}
Define $L_{1,k} := L_1+\dots + L_k$ and observe that $I_P(L_{1,k}; S)\leq I_P(L_1,\dots,L_k; S)$. Hence,  
\begin{equation}\label{eq:multi_suff2}
\begin{split}
\theta_{L,k}(R) :=  \frac 1 k \bigg(\sup_{p_{S|L_1,\dots L_k}} & E_P[S]- \log E_P\Big[\exp\Big(S-L_{1,k}\Big)\Big]\bigg)\\
\text{s.t.}& \quad I_P(L_{1,k}; S)\leq kR
\end{split}
\end{equation}
is an upper bound to $\tilde\theta_{L,k}(R)$ as the optimization domain is enlarged (it is not difficult to show that $\theta_{L,k}$ is in fact equal to $\tilde\theta_{L,k}$). Since both the objective and constraint functions in \eqref{eq:multi_suff2} only depend on $L_{1,k}$, the feasible set can be reduced to the set of channels from $L_{1,k}$ to $S$. Hence, 
\begin{equation}\label{eq:multi_suff3}
\begin{split}
\theta_{L,k}(R) =  \frac 1 k \bigg(\sup_{p_{S|L_{1,k}}} & E_P[S]- \log E_P\Big[\exp\Big(S-L_{1,k}\Big)\Big]\bigg)\\
\text{s.t.}& \quad I_P(L_{1,k}; S)\leq kR.
\end{split}
\end{equation}

	Note the resemblance of \eqref{eq:multi_suff3} to \eqref{eqn:curve}. Consequently, all results for one-dimensional quantization directly translate to the multi-dimensional case and we obtain the following upper bound to the boundary of the achievable region:
	\begin{equation}\label{eqn:multi_curve}
	\begin{split}
	\theta_{U,k}(R) :=  \frac 1 k \bigg(\sup_{p_{V|U}} &\quad E_P[V]- \log E_P[\exp(V-U_k)]\bigg)\\
	\text{s.t.}& \quad I_P(U_k; V)\leq kR\\
	=  \frac 1 k \theta_{U_k}(k&R)
	\end{split}
	\end{equation}
	where $U_k$ is the random variable that has the same distribution as $L_{1,k}$. Following the same steps we have taken for the one-dimensional case, we can also obtain the gap function and the rate-gap curve for the $k$-dimensional case as
	\begin{equation}\label{eqn:multi_curve2}
	\delta_{U,k}(R) =  \frac 1 k \delta_{U_k}(kR),\quad R_{U,k}(\delta) = \frac 1 k R_{U_k}(k\delta).
	\end{equation}
	The previously obtained upper and lower bounds for the one-dimensional case are therefore valid for $k$-dimensional case as well:
	\begin{equation}
	\uR_{U,k}(\delta) := \frac 1 k \uR_{U_k}(k\delta) \leq R_{U,k}(\delta) \leq \frac 1 k \oR_{U_k}(k\delta) =: \oR_{U,k}(\delta).
	\end{equation}
	For various $k$ values, the lower bounds $\uR_{U,k}(\delta)$ and upper bounds $\oR_{U,k}(\delta)$ are drawn in Figure \ref{fig:multiple} for the same scenario in Figure \ref{fig:comparison}.
	\begin{figure}[h!]
		\centering
		\includegraphics[scale=0.65]{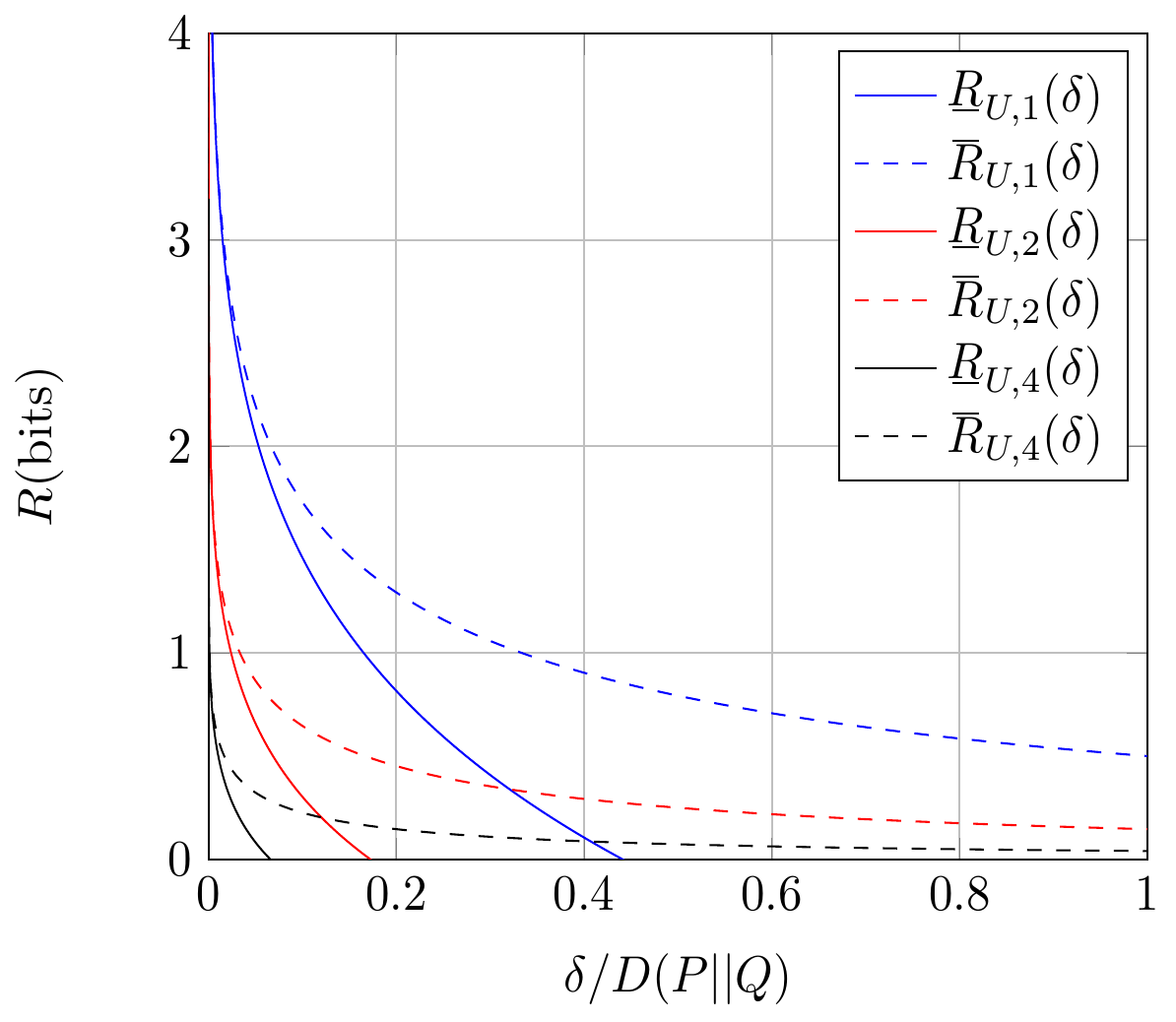}
		\caption{Upper and lower bounds under $k$-dimensional vector quantization, for $k = 1,2,4$. Again, the same setup in Figure \ref{fig:comparison} is considered. The lower bounds $\uR_{U,k}$ are drawn as solid curves and the upper bounds $\oR_{U,k}$ are drawn as dashed curves. Although all upper and lower bounds are pointwise decreasing with $k$, it is not certain that the true $R_{U,k}$ curves exhibit the same behavior.}
		\label{fig:multiple}
	\end{figure}
	
	Observe that $R_{U,k}$'s obey the subadditive relation
	\begin{equation}
	(k+l)R_{U,k+l}(\delta) \leq k R_{U,k}(\delta) + l R_{U,l}(\delta)
	\end{equation}
	as the admissible strategies for the quantization of $k+l$ samples include the strategies that quantize $k$ samples and $l$ samples separately. However, note that this does not imply $R_{U,k}(\delta) \leq R_{U,l}(\delta)$ for $k \geq l$.  Nevertheless, from a well-known result on subadditive sequences, e.g. \cite{Gallager}, we know 
	\begin{equation}
	\lim_{k \to \infty} R_{U,k}(\delta) = \inf_k R_{U,k}(\delta)
	\end{equation}
	and using the upper bound \eqref{eq:upper_bnd}, we obtain for $\delta > 0$,
	\begin{equation}\label{eq:ratetozero}
	\begin{split}
	\lim_{k \to \infty} R_{U,k}(\delta) &\leq \lim_{k \to \infty} \frac 1 {2k} \log\bigg(1+\frac{\Var(U_k)}{2k\delta}\bigg)\\
	&= \lim_{k \to \infty} \frac 1 {2k} \log\bigg(1+\frac{\Var(U_1)}{2\delta}\bigg)\\
	&= 0.
	\end{split}
	\end{equation}
	This is in contrast with the classical rate-distortion function as it is already defined for $k \to \infty$. 
	
	Although \eqref{eq:ratetozero} shows that the lower bound $R_{U,k}$ tends to zero, this is also true for the true boundary curve.  A simple achievability scheme at large $k$ is as follows: Since the remote node records the data until $k$, it can make its own decision $\cH_0$ or $\cH_1$ and send the one-bit result to the fusion center. The average number of bits sent is then kept arbitrarily small and since the node makes the estimate of the true hypothesis based on an optimal test, the type-II error rate will be close to $D(P||Q)$, which is the best possible decay rate. Although such a design might seem appealing in terms of the performance of type-II error rate, the peripheral node needs to have sufficient computational power as a requirement of this design. Also recall that in the end of Section \ref{sec:high_rate}, we mentioned that the covering efficiency of lattices may improve at high dimensions. However, \eqref{eq:ratetozero} and the strategy we have just described suggest that there is no need for lattice quantization for high dimensions --- the node only sends its one-bit decision.
	
	\section{Multiple-Node Case}\label{sec:multiple_node}
	
	All the previous results obtained for the single-node case can be extended to the multiple-node case. This is due to the fact that the data is \emph{independent across nodes}. To make this extension, we provide a modified definition of achievable pairs. Recall that at time $t$, node $i$ observes data coming from $P^{(i)}$ under $\cH_0$, and from $Q^{(i)}$ under $\cH_1$; calculates the LLR $L_t^{(i)}$, and compresses it with a simple function $S^{(i)}_t = f_t^{(i)}(L_t^{(i)})$. Furthermore, as discussed before, if 
	\begin{equation}
	\frac 1 t \sum_{\tau = 1}^t H_P(S_\tau^{(i)}) \leq R_i/\log_2e,
	\end{equation}
	then the compressed scores can be sent losslessly with an average number of bits less than $R_i$. After recalling the system dynamics, we provide the modified version of Definition \ref{def:achievable_single}.
	
	\begin{definition}Given $\{P^{(i)}\}_{i=1}^m$ and $\{Q^{(i)}\}_{i=1}^m$, $(R_1,\dots,R_m,\theta)$ is an achievable pair if there exists $m$ sequences $\{f_t^{(1)}\},\dots,\{f_t^{(m)}\}$ of simple functions and a sequence of thresholds $\{\eta_t\}$ such that 
		\begin{itemize}
			\item[(a)] $\frac 1 t \sum_{\tau = 1}^t H_P(S_\tau^{(i)}) \leq R_i$, for all $t$ and for all $i$
			\item[(b)] $\lim_{t \to \infty}\alpha_t = 0$
			\item[(c)] $\liminf_{t \to \infty}\frac 1 t \log \frac 1 {\beta_t} \geq \theta$
		\end{itemize}
		where $S_t^{(i)} = f_t^{(i)}(L_t^{(i)})$, and $\alpha_t$, $\beta_t$ are the type-I and type-II errors respectively, defined in \eqref{eq:errors_def}.\hfill\qedsymbol
	\end{definition}
	
	Let $\theta_t(R_1,\dots,R_m) := \sum_{i=1}^m\theta_t^{(i)}(R_i)$, where 
	\begin{equation}
	\begin{split}
	\theta^{(i)}_t(R_i) := \sup_{\substack{\{f^{(i)}_1,\dots,f^{(i)}_t\}\\ \in \cF_t(R_i)}} \frac 1 t \sum_{\tau = 1}^t \bigg(E_P[S^{(i)}_\tau]-\log E_P[e^{S^{(i)}_\tau-L^{(i)}_\tau}]\bigg)
	\end{split}
	\end{equation}
	is defined as in Theorem \ref{thm:theta}. Observe that
	\begin{equation}
	\begin{split}
	\theta_t(R_1,\dots,R_m) &= \sum_{i=1}^m\theta_t^{(i)}(R_i)\\
	&= \sum_{i=1}^m\sup_{\substack{\{f^{(i)}_1,\dots,f^{(i)}_t\}\\ \in \cF_t(R_i)}} \frac 1 t \sum_{\tau = 1}^t \bigg(E_P[S^{(i)}_\tau]-\log E_P[e^{S^{(i)}_\tau-L^{(i)}_\tau}]\bigg)\\
	&=\sup_{\substack{\{f^{(i)}_1,\dots,f^{(i)}_t\}\\ \in \cF_t(R_i),\ i\leq m}}\frac 1 t \sum_{\tau = 1}^t\sum_{i=1}^m\bigg(E_P[S^{(i)}_\tau]-\log E_P[e^{S^{(i)}_\tau-L^{(i)}_\tau}]\bigg)
	\end{split}
	\end{equation}
	and 
	\begin{equation}
	\begin{split}
	\sum_{i=1}^m\log E_P[e^{S^{(i)}_\tau-L^{(i)}_\tau}] &\stackrel{(a)}{=} \log E_P[e^{\sum_{i=1}^m (S^{(i)}_\tau-L^{(i)}_\tau)}] \\
	&\stackrel{(b)}{=}  \log E_Q[e^{\sum_{i=1}^m S^{(i)}_\tau}]\\
		&\stackrel{(c)}{=}  \sum_{i=1}^m\log E_Q[e^{ S^{(i)}_\tau}]
	\end{split}
	\end{equation}
	where $(a)$, $(c)$ are due to the independence assumption across nodes and $(b)$ is because $e^{-\sum_{i=1}^m L^{(i)}_\tau} = \prod_{i=1}^m \frac{dP^{(i)}}{dQ^{(i)}}=\frac{dP}{dQ}$ is the Radon--Nikodym derivative of $P$ with respect to $Q$. Therefore, we have
	\begin{equation}\label{eq:theta_multi_node}
	\begin{split}
	\theta_t(R_1,\dots,R_m) = \sum_{i=1}^m\sup_{\substack{\{f^{(i)}_1,\dots,f^{(i)}_t\}\\ \in \cF_t(R_i)}}\frac 1 t \sum_{\tau = 1}^t\bigg(E_P[S^{(i)}_\tau]-\log E_Q[e^{S^{(i)}_\tau}]\bigg).
	\end{split}
	\end{equation}
	Comparing \eqref{eq:theta_multi_node} with \eqref{eq:theta_single_node} and following exactly the same steps in the proof of Theorem \ref{thm:theta}, we obtain the analogous version of Theorem \ref{thm:theta} (ii):
	\begin{equation}\label{eq:multiple_node_theta}
	\begin{split}
	\theta^*(R_1,\dots,R_m) :=& \sup\{\theta:(R_1,\dots,R_m,\theta) \text{ achievable}\}\\
	=&\lim_{t \to \infty} \theta_t(R_1,\dots,R_m)\\
	=&\sum_{i=1}^m \lim_{t\to \infty} \theta_t^{(i)}(R_i)\\
	=&\sum_{i=1}^m \breve\theta_1^{(i)}(R_i),
	\end{split}
	\end{equation}
	which characterizes the boundary of the optimal curve for rate constraints $(R_1,\dots,R_m)$. We also know that $\breve\theta_1(R)$ is upper bounded by $\theta_U(R) = D(P||Q) - \delta_U(R)$ from Corollary \ref{cor:upper_bnd}, hence
	\begin{corollary}
		\begin{equation}
		\theta^*(R_1,\dots,R_m) \leq \sum_{i=1}^m \Big[D(P^{(i)}||Q^{(i)}) - \delta_{U_i}^{(i)}(R_i) \Big]
		\end{equation}
		where $U_i$ has the same distribution as $L_1^{(i)}$, LLR of node $i$.\hfill\qedsymbol
	\end{corollary}
	
	One might also consider an extension of the problem to sum-rate constraints. Namely, the communication constraint is redefined as $R_1 + \dots  + R_m \leq R_{\text{sum}}$. The characterization in \eqref{eq:multiple_node_theta} is readily adapted to sum-rate constraint as
	\begin{equation}
	\begin{split}
	\theta^*(R_{\text{sum}}) := \sup\{&\theta:(R_1,\dots,R_m,\theta) \text{ achievable},\ R_1+\dots+R_m= R_{\text{sum}}\}\\
	&\kern-33pt=\max_{R_1+\dots + R_m = R_{\text{sum}}}\sum_{i=1}^m \breve\theta_1^{(i)}(R_i),
	\end{split}
	\end{equation}
	and
	\begin{equation}\label{eq:sum_rate_upper}
	\theta^*(R_{\text{sum}})\leq \max_{R_1+\dots + R_m = R_{\text{sum}}}\sum_{i=1}^m \theta_{U_i}(R_i).
	\end{equation}
	The $(\theta, R)$ pairs that lie above the curve on right-hand side are unachievable under the sum-rate constraint. Hence, one may be interested in the optimal rate sharing that maximizes the right-hand side with an aim to characterize an unachievable region. We first provide a simple property of a possible optimal allocation: Intuitively, the optimal rate sharing must not exclude the more informative nodes.
	\begin{definition}[\kern-3pt\cite{Blackwell}]
		A node $j$ is said to be more informative\footnote{Although the definition of more informativeness is different for $m$-ary hypothesis tests, it is shown in \cite{Blackwell} that for $m=2$, the definition given here is equivalent.} than node $i$, and denoted as $i \prec j$ if there exists a probability transition kernel $w:\cB(\mathbb{R})\times\mathbb{R} \to \mathbb{R}_+$ such that for all $\cA \in \cB(\mathbb{R})$
		\begin{equation}
		\int w(\cA,x) dP^{(j)}(x) = P^{(i)}(\cA)
		\end{equation}
		and 
		\begin{equation}
		\int w(\cA,x) dQ^{(j)}(x) = Q^{(i)}(\cA).
		\end{equation}\hfill\qedsymbol
	\end{definition}
	
	With the above definition, we have 
	\begin{theorem}\label{thm:more_informative}
		Let $\cR^* \subset \mathbb{R}^m$ be the set of optimal allocations in \eqref{eq:sum_rate_upper}. If $i \prec j$ for some $j$, then $\cR^*$ contains a $(R_1^*,\dots,R_m^*)$ such that $R_i^* > 0$ only if $R_j^* > 0$.
	\end{theorem}
	\begin{proof}
		Recall the definition of $\theta_{U_i}(R_i)$:
		\begin{equation}\label{eq:dominate}
		\begin{split}
		\theta_{U_j}(R) &= \sup_{\substack{p_{V|U_j} \\ \text{s.t.}\  I(U_j; V)\leq R}} \quad E_P[V]- \log E_P[\exp(V-U_j)]\\
		&\stackrel{(a)}{=} \sup_{\substack{p_{V|U_j} \\ \text{s.t.}\  I(U_j; V)\leq R}} \quad E_P[V]- \log E_Q[\exp(V)]\\
		&\geq \sup_{\substack{p_{V|U_i}\circ\, w_{U_i|U_j} \\ \text{s.t.}\  I(U_j; V)\leq R}}\quad E_P[V]- \log E_Q[\exp(V)]\\
		&\stackrel{(b)}{\geq} \sup_{\substack{p_{V|U_i}\circ\, w_{U_i|U_j} \\ \text{s.t.}\  I(U_i; V)\leq R}} \quad E_P[V]- \log E_Q[\exp(V)]\\
		&= \theta_{U_i}(R)
		\end{split}
		\end{equation}
		where $(a)$ is due to the measure change as $U_j$ is distributed as the logarithm of the Radon--Nikodym derivative $\frac{dP^{(j)}}{dQ^{(j)}}$ and $(b)$ is due to the data processing inequality $I(U_j;V) \leq I(U_i;V)$ for the choice of $p_{V|U_j} = p_{V|U_i}\circ w_{U_i|U_j}$. Suppose $R_i > 0$ and $R_j = 0$. Since $\theta_{U_j}(R)$ pointwise dominates $\theta_{U_i}(R)$, one cannot do worse with the modification $(R_i,0) \to (0,R_i)$.
	\end{proof}
	
	It might be tempting to think that the optimal allocation assigns $R_i^* = 0$ if $i \prec j $ for some $j$, as $\theta_{U_j}(R)$ dominates $\theta_{U_i}(R)$ pointwise. This is not true in general. Suppose $P^{(i)}$ and $Q^{(i)}$ are obtained by passing $P^{(j)}$ and $Q^{(j)}$ through an additive Gaussian channel with almost zero noise and suppose $R_{\text{sum}}$ is very large. Allocating all the rate to node $j$ will yield an exponent close to $D(P^{(i)}||Q^{(i)})$ whereas an equal rate allocation gives an exponent close to $D(P^{(i)}||Q^{(i)}) + D(P^{(j)}||Q^{(j)})$. Therefore, if the sum-rate constraint is large enough, it is preferred to observe two (almost uncompressed) independent samples instead of one, which surely increases the type-II decay rate.
	
	Since $\theta_{U_i}$'s (or equivalently $\delta_{U_i}$'s) are difficult to calculate in general, one may consider the optimal allocation based on the $(R, \delta)$ pairs that lie on the $\uR_{U_i}(\delta)$ curves given by the parametric form \eqref{eqn:lower_gap}. 
	Observe that such pairs depend on $U_i$ only through shifts of a parametric curve by its differential entropy $h_i := h(U_i)$. Denoting the inverse of $\uR_{U_i}(\delta)$ by $\underline{\delta}_{U_i}(R)$, we have thus the property 
	\begin{equation}
	\underline{\delta}_{U_i}(R) = \underline{\delta}(R - h_i) 
	\end{equation}
	where $\underline{\delta}(R)$ is given by the parametric form
	\begin{equation}\label{eq:parametric_theta}
	\begin{split}
	R(\alpha) &= - \log \Gamma(\alpha) + \alpha \psi(\alpha) - \alpha,\\
	\underline{\delta}(\alpha) &=  \log \alpha-\psi(\alpha),\quad \alpha > 0.
	\end{split}
	\end{equation}
	The sum-rate optimization is then formulated as 
	\begin{equation}\label{eq:waterfill}
	\begin{split}
	\underline{\theta}^*(R_{\text{sum}}) &:= \max_{\sum_{i=1}^m R_i = R_{\text{sum}}} \sum_{i=1}^m \big[D(P^{(i)}||Q^{(i)})-\underline{\delta}(R_i - h_i)\big]\\
	&=\sum_{i=1}^m D(P^{(i)}||Q^{(i)}) - \min_{\sum_{i=1}^m R_i = R_{\text{sum}}} \sum_{i=1}^m\underline{\delta}(R_i - h_i),
	\end{split}
	\end{equation}
	which can be shown to admit a water-filling solution.
	\begin{lemma}\label{lem:waterfill} The sum-rate constrained problem \eqref{eq:waterfill} has a solution given by 
		\begin{equation}\label{eq:waterfill_soln}
		R_i^* = (\mu + h_i)^+,
		\end{equation}
		where $(x)^+$ denotes the positive part of $x$; and $\mu$ is a constant chosen to satisfy the sum-rate constraint $\sum_{i=1}^m R^*_i = R_{\text{sum}}$.
	\end{lemma}
	\begin{proof}
%
		First, observe that $\underline{\delta}(R)$ is convex. This is a consequence of its formulation in \eqref{eq:max_ent}.
		As $\underline{\delta}(R)$ is convex, the Karush--Kuhn--Tucker (KKT) conditions are necessary and sufficient to characterize the solutions. For the sum-constraint $\sum_{i=1}^m R_i = R_{\text{sum}}$, it is known that the KKT conditions are given by \cite{Gallager}
		\begin{equation}
		\begin{split}
		\underline{\delta}'(R_i-h_i) = \lambda ,\quad R_i > 0\\
		\underline{\delta}'(-h_i) > \lambda ,\quad R_i = 0
		\end{split}
		\end{equation}
		for some constant $\lambda$. Since $\underline{\delta}'(R_i-h_i)$ is non-decreasing due to the convexity of $\underline{\delta}(R)$, the KKT conditions are also equivalent to
		\begin{equation}
		\begin{split}
		R_i-h_i = \mu ,\quad R_i > 0\\
		-h_i > \mu ,\quad R_i = 0
		\end{split}
		\end{equation}
		with $\mu = R(\alpha)|_{\alpha = 1/\lambda}$, and they characterize the claimed solution in \eqref{eq:waterfill_soln}.
	\end{proof}
	
	Although the sum-rate optimization in \eqref{eq:waterfill} does not give the exact boundary of the achievable $(R_1,\dots,R_m,\theta)$ pairs,  $\underline{\theta}^*(R_{\text{sum}})$ is an upper bound to the boundary. This implies that no $(R_1,\dots,R_m,\theta)$ pair lying above $\underline{\theta}^*(R_{\text{sum}})$ is achievable. A numerical example is illustrated in Figure \ref{fig:waterfilling}.

	\begin{figure}[h!]\label{fig:waterfilling}\centering
		\includegraphics[scale=0.65]{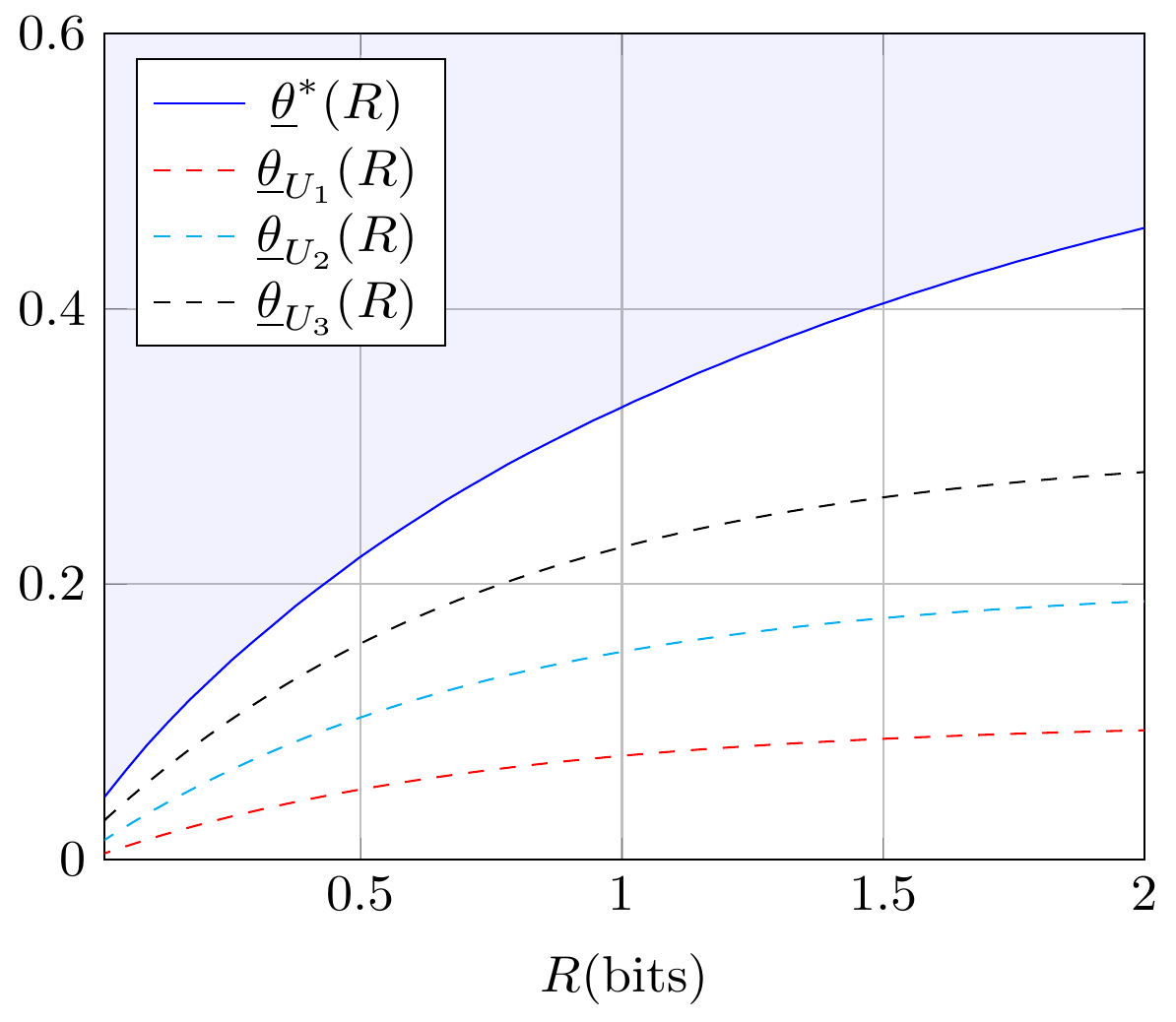}
		\caption{A 3-node instance of the problem. Nodes observe zero-mean Gaussian data under $\cH_0$, and with a mean vector $[\sqrt{0.2}, \sqrt{0.4},\sqrt{0.6}]$ under $\cH_1$. The data has unit variance under both hypotheses. Then, the LLRs also have Gaussian distributions with means $[0.1,0.2,0.3]$ and variances $[0.2,0.4,0.6]$ under $\cH_0$. The individual $\underline\theta_{U_i}(R) := D(P^{(i)}||Q^{(i)}) - \underline\delta_{U_i}(R)$ curves are drawn dashed, whereas the optimal sum-rate curve $\underline\theta^*(R)$ is drawn solid. The shaded region is unachievable under sum-rate constraints.}
	\end{figure}
	
	We conclude this section by noting that when the data is not independent across nodes, the question of how to combine the scores is highly non-trivial even if there were no communication constraints. Under communication constraints, the problem for this general case could be of formidable complexity.
	
	\section{Discussion}
	
	In this work, we have studied a fundamental limit of a distributed hypothesis testing problem when remote nodes compress their data in a memoryless fashion and the expected number of bits sent under $\cH_0$ should be kept limited to a prescribed quantity $R$. This asymmetric communication constraint is in line with the view that $\cH_1$ is a rare high-risk event and must be detected with high probability. Thus, nodes are allowed to send a large number of bits under $\cH_1$. With such a communication constraint, we characterized the maximum attainable type-II error (i.e., mis-detection of $\cH_1$) exponent (Theorem \ref{thm:theta}) for vanishing type-I error probability and derived a closed-form upper bound to this error exponent (Lemma \ref{lem:parametric}).
	
	 In the high-rate regime, we show that the upper bound is approached with simple scalar lattice quantization within $\frac 1 2\log_2(\pi e/2) \approx 1.047$ bits. This gap is due to the covering inefficiency of the 1-dimensional lattice and due to the fact that the gap to the optimal error rate $D(P||Q)$ behaves quadratically at high rates. Therefore, it is expected that the results for the high-rate regime coincide with the results on the rate-distortion problem for mean-square distortion. This is the reason that the asymptotic behavior of the rate-gap curve (Theorem \ref{thm:asymp_rate}) is reminiscent of the Shannon lower bound for the rate-distortion curve under mean-square distortion. It is also because of this quadratic behavior that one can approach the lower bound within $\frac 1 2\log_2(\pi e/2)$ bits under scalar lattice quantization.
	
	We have also obtained a simple upper bound for the vector quantization case that can be expressed in terms of its scalar quantization analog \eqref{eqn:multi_curve}--\eqref{eqn:multi_curve2}. Hence, the same upper and lower bounds for the scalar case are also valid for the vector quantization case. We have also shown in \eqref{eq:ratetozero} that as the dimension tends to infinity, the rate-gap curve is identically equal to zero for $\delta > 0$. This is consistent with the following simple achievability scheme: The remote node performs its own Neyman--Pearson test and sends its decision with one bit. Hence, the fusion center is informed of the optimal decision, and the average number of bits sent is arbitrarily small.
	
	With an independence assumption across nodes, the results for the single-node problem can be easily extended to the multiple-node problem, and a simple upper bound to the optimal type-II error exponent can be obtained in terms of the sum of individual upper bounds for each remote node. We formulated a sum-rate constrained problem and studied some of its properties (Theorem \ref{thm:more_informative} and Lemma \ref{lem:waterfill}).
	
	As a final remark, we note that the results for the vector quantization case are also applicable to the multiple-node case. This implies that when the dimension tends to infinity, the rate-gap curve will be again equal to zero and is attained with a simple scheme that is similar to the single-node case: Each node performs its optimal test and sends the 1-bit result to the fusion center. The center decides $\cH_1$ if at least one node decides $\cH_1$. This scheme ensures vanishing type-I error probability and the type-II error exponent is equal to $\sum_i D(P^{(i)}||Q^{(i)})$. Hence the center is able to attain the optimal rate with an arbitrarily low amount of communication. However, this scheme allows each node to dictate a $\cH_1$ decision to the center. This results in a system that is vulnerable to manipulation.  By contrast, schemes with scalar quantization, or with low-dimensional vector quantization, give the center the opportunity to detect errors or manipulations and therefore these schemes could be of interest when faulty or malicious nodes are present.
	
	\bibliographystyle{IEEEtran}
	\bibliography{ref}
	\appendix
	
	\subsection{Proof of Theorem 1}\label{app:theta} 
	\subsubsection{Proof of (i)} Let $\breve\theta_1(R)$ be the concave envelope of $\theta_1(R)$. Recall $S_t = f_t(L_t)$ and since all the expectations are taken under $P$, we omit $P$ from the subscripts. Also recall \begin{equation*}\theta_t(R) = \sup_{\{f_1,\dots,f_t\} \in \cF_t(R)} \frac 1 t \sum_{\tau = 1}^t \bigg(E[S_\tau]-\log E[e^{S_\tau-L_\tau}]\bigg).
	\end{equation*}
	We first show that $\breve\theta_1(R) \geq \theta_t(R)$ for all $t$. Let us modify the definition of $\theta_1(R)$ as 
	\begin{equation}\label{eqn:theta_1}
	\begin{split}
	\theta_1(R) = \sup_{f\text{ simple}} &\quad E[S_1]\\
	\text{s.t.} &\quad E[e^{S_1-L_1}] = 1\\
	& \quad H(S_1) \leq R
	\end{split}
	\end{equation}
	since shifting $S_1$ does not change the entropy. The supremization is over simple functions. Similarly, $\theta_t(R)$ can be defined as
	\begin{equation}\label{eqn_theta_t}
	\begin{split}
	\theta_t(R) = \sup_{\substack{f_1,\dots,f_t\\ \text{simple}}} &\quad\frac 1 t \sum_{i = 1}^t  E[S_i]\\
	\text{s.t.} &\quad E[e^{S_i-L_i}] = 1, \forall i \leq t\\
	& \quad \frac 1 t \sum_{i = 1}^t H(S_i) \leq R.
	\end{split}
	\end{equation}
	For any $\{f_i\}$ in the feasible set of \eqref{eqn_theta_t}, there exists $\{R_i\}$ with $\frac 1 t \sum_{i = 1}^t R_i \leq R$ and $H(S_i) \leq R_i$ for all $i \leq t$; and consequently $\frac 1 t \sum E[S_i] \leq \frac 1 t \sum \theta_1(R_i) \leq \breve{\theta}_1(R)$. Thus, $\theta_t(R) \leq \breve \theta_1(R)$.

	It remains to prove the reversed inequality for $t\to \infty$, i.e., $\theta_t(R) \geq \breve\theta_1(R)-\epsilon$ for large enough $t$, given $\epsilon  >0$. Suppose $\theta_1(R)$ is attained in the limit of the sequence of simple functions $\{f_t^*\}$. This implies for all $\epsilon_1 > 0$, there exists a simple function $f$ that maps $L_1 \mapsto S_1$ such that $E[S_1] \geq \theta_1(R) - \epsilon_1$, $E[e^{S_1-L_1}] = 1$ and $H(S_1) \leq R$. Carathéodory's theorem \cite[Section 17]{rockafellar-1970a} ensures that every point on the concave envelope $\breve\theta_1(R)$ is achieved by a convex combination of at most two points on $\theta_1(R)$. This implies the existence of functions $f$, $\tilde f$, and $\lambda \in [0,1]$, such that $\lambda E[S_1] + (1-\lambda) E[\tilde S_1] \geq \breve\theta_1(R) - \epsilon_2$ for all $\epsilon_2 > 0$, and $\lambda H(S_1) + (1-\lambda)H(\tilde S_1) \leq R$. Assume $H(\tilde S_1) \leq R$ without loss of generality. Consider the sequence $\{f_t\}$ such that $f_i = f$ for $i \leq \lceil\lambda t\rceil$ and $f_i = \tilde f $ otherwise. Observe $\frac 1 t \sum H(S_i) \leq R$ and thus $\theta_t(R) \geq \frac 1 t \sum E[S_i] \geq \breve\theta_1(R) - 2\epsilon_2$ for $t$ large enough. The proof of part (i) is complete.
	
	%
	
	In the achievability proof of part (ii), we have to show that the supremizers of $\theta_t(R)$ have to drive $\alpha_t \to 0$. For completeness, we provide the proof here. Take $\{f_t\}$ as above and use Chebyshev's inequality to bound the type-I error probability under the threshold test with the threshold $\eta_t$ chosen as in \eqref{eqn:threshold}:
	\begin{equation}
	\alpha_t = P\bigg(\sum_{i = 1}^t f_i(L_i) < E_P[g_i(L_i)]-\epsilon\bigg) \leq \frac{\sum_{i=1}^t \text{Var}\big(f_i(L_i)\big)}{\epsilon^2 t^2}
	\end{equation}
	Recall that $f_i$ is defined to be equal either to $f$ or $\tilde f$. As $f$ and $\tilde f$ take finitely many values, the variances of $f(U)$, $\tilde f(U)$ are bounded for any $U$. Therefore,
	\begin{equation}
	\alpha_t \leq \frac{\max\Big\{\text{Var}\big(f(L_1)\big),\text{Var}\big(\tilde f(L_1)\big)\Big\}}{\epsilon^2 t} \to 0,
	\end{equation}
	which shows that any sequence in the achievability part of (i) indeed satisfies the property (b) in Definition 1.
	
	\subsubsection{Proof of (ii)} (Achievability)
	As mentioned, choose $\eta_t = \frac 1 t \sum_{\tau = 1}^t E_P[S_t] -\epsilon$ as in \eqref{eqn:threshold}. We upper bound the type-II error $\beta_t$ as 
	\begin{equation}
	\begin{split}
	Q\big(\bar S_t \geq \eta_t\big) &= Q\big(\exp(t\bar S_t) \geq \exp\big(t\eta_t\big)\big)\\
	&\stackrel{(a)}{\leq} E_Q\bigg[\exp\Big(\sum_{\tau = 1}^t \big(S_\tau-E_P[S_\tau]+ \epsilon\big)\Big)\bigg]\\
	& \stackrel{(b)}{=} \prod_{\tau=1}^tE_Q[e^{S_\tau}]\exp(-E_P[S_\tau]+\epsilon)\\
	\end{split}
	\end{equation}
	where $(a)$ follows from Markov inequality and $(b)$ follows from independent processing of LLRs.
	Therefore,
	\begin{equation}
	\frac 1 t \log \frac 1 {\beta_t} \geq \frac 1 t \sum_{\tau = 1}^t \big(E_P[S_\tau]-\log E_Q[e^{S_\tau}]\big) -\epsilon.
	\end{equation}
	Optimizing the right-hand side with respect to the choice of $f_t$'s satisfying the communication constraints we have
	\begin{equation}\label{eqn:lower_bnd_donsker}
	\frac 1 t \log \frac 1 {\beta_t} \geq \sup_{\{f_1,\dots,f_t\} \in \cF_t(R)} \frac 1 t \sum_{\tau = 1}^t E_P[S_\tau]-\log E_Q[e^{S_\tau}] -\epsilon.
	\end{equation}
	Consider the transformation $\tilde S_t = \log \frac{P(S_t)}{Q(S_t)}$, i.e., the LLR of $S_t$. Observe that the mapping $L_t \mapsto \tilde S_t$ is a simple function and since $S_t$ is discrete, $H(S_t) \geq H(\tilde S_t)$ as the mapping $S_t \mapsto \tilde S_t$ is deterministic. Therefore, communication constraints are still satisfied. Furthermore, $\tilde S_t$ is a sufficient statistic and the fusion center is able to deploy a Neyman--Pearson test based on $\tilde S_t$'s.
	It is known from Donsker--Varadhan representation \cite{polyanskiy2014lecture} of divergence that 
	\begin{equation}\label{eqn:Donsker_varadhan}
	D(P||Q) = \sup_{g:\mathbb{R} \to \mathbb{R}} E_P[g(X)] - \log E_Q[e^{g(X)}]
	\end{equation}
	where the supremum is over the set of bounded measurable functions on $\mathbb{R}$, and is attained at $g(X) = \log\frac{dP}{dQ}$, the logarithm of the Radon--Nikodym derivative of $P$ with respect to $Q$. We combine \eqref{eqn:Donsker_varadhan} with \eqref{eqn:lower_bnd_donsker} to obtain
	\begin{equation}\label{eqn: lower_bnd}
	\frac 1 t \log \frac 1 {\beta_t} \geq  \sup_{\{f_1,\dots,f_t\} \in \cF_t(R)}  \frac 1 t \sum_{\tau = 1}^t  D(P_{f_\tau}||Q_{f_\tau})- \epsilon,
	\end{equation}
	where $P_{f_\tau} := P \circ f_\tau$ and $Q_{f_\tau} := Q \circ f_\tau$. It only remains to show that the supremizers in \eqref{eqn: lower_bnd} must drive $\alpha_t \to 0$, which we have already proved at the end of part (i).\\
	
	(Converse)
	Now, following similar steps to Stein's lemma, we apply data processing inequality twice to see that for any sequence of $f_t$'s:
	\begin{equation}
	tD(P||Q) \geq \sum_{\tau = 1}^t D(P_{f_\tau}||Q_{f_\tau}) \geq d(\alpha_t || 1-\beta_t)
	\end{equation}
	where $d(p||q) := p\log \frac{p}{q} + \bar p\log \frac{\bar p}{\bar q}$ is the binary divergence. Hence,
	\begin{equation}
	\sum_{\tau = 1}^t D(P_{f_\tau}||Q_{f_\tau}) \geq -h_e(\alpha_t) - \alpha_t \log(1-\beta_t) - (1-\alpha_t)\log(\beta_t),
	\end{equation}
	with $h_e(p) := - p\log p - (1-p)\log(1-p)$.
	Suppose $\alpha_t \to 0$ and $\beta_t$ bounded away from 1. Then, it must be true that
	\begin{equation}\label{eqn: upper_bnd}
	\liminf_{t \to \infty}  \sup_{\{f_1,\dots,f_t\} \in \cF_t(R)} \frac 1 t \sum_{\tau = 1}^t D(P_{f_\tau}||Q_{f_\tau}) \geq \liminf_{t \to \infty}\frac 1 t \log \frac 1 {\beta_t}.
	\end{equation}
	Taking $\liminf$ and $\epsilon \to 0$ in \eqref{eqn: lower_bnd}, combining with \eqref{eqn: upper_bnd}; we therefore have 
	\begin{equation}
	\liminf_{t \to \infty} \sup_{\{f_1,\dots,f_t\} \in \cF_t(R)}  \frac 1 t \sum_{\tau = 1}^t \sup_{f_\tau} D(P_{f_\tau}||Q_{f_\tau}) = \liminf_{t \to \infty} \frac 1 t \log \frac 1 {\beta_t}.
	\end{equation}
	In other words, $\beta_t$ decays with an exponent at least
	\begin{equation}\liminf_{t \to \infty}\sup_{\{f_1,\dots,f_t\} \in \cF_t(R)}   \frac 1 t \sum_{\tau = 1}^t D(P_{f_\tau}||Q_{f_\tau}).\end{equation}
	Recall that the set $\cF_t(R)$ also includes the supremal function $g_t$'s in the Donsker--Varadhan formulation \eqref{eqn:Donsker_varadhan}, thus we have
	\begin{equation}
	\begin{split}
	& \liminf_{t \to \infty} \frac 1 t \log \frac 1 {\beta_t}\\ & =\liminf_{t \to \infty} \sup_{\{f_1,\dots,f_t\} \in \cF_t(R)}  \frac 1 t \sum_{\tau = 1}^t E_P[S_\tau]-\log E_Q[e^{S_\tau}]\\
	& = \liminf_{t \to \infty} \sup_{\{f_1,\dots,f_t\} \in \cF_t(R)}  \frac 1 t \sum_{\tau = 1}^t E_P[S_\tau]-\log E_P[e^{S_\tau-L_\tau}]\\
	& = \lim_{t \to \infty} \theta_t(R).
	\end{split}
	\end{equation}
	
	\subsection{Proof of Lemma \ref{lem:dist_rate}}\label{app:dist_rate}
	We consider the inequality \eqref{eqn:lower_bnd_donsker} and obtain a more relaxed lower bound for it using $\log x \leq x-1$ as
	\begin{equation}
	\frac 1 t \log \frac 1 {\beta_t} \geq \frac 1 t \sum_{\tau = 1}^t \sup_{f_\tau \in \cF} E_P[S_\tau]- E_Q[e^{S_\tau}] + 1-\epsilon.
	\end{equation}
	It is known that $D(P||Q)$ can also be represented as
	\begin{equation}
	D(P||Q) = \sup_{g:\mathbb{R} \to \mathbb{R}} E_P[g(X)] - E_Q[e^{g(X)}]+1.
	\end{equation}
	Proceeding similarly to the proof of Theorem \ref{thm:theta}, we obtain $\theta_X(R) = \tilde \theta_X(R)$.
	
	\subsection{Proof of Lemma \ref{lem:delta_limit}}\label{app:delta_limit}
	Consider the following quantization of $U$. Let $\cI_n := [-n,n]$ and $v_{n,k} := -n+k2^{-n}$ for $1\leq k\leq 2^{n+1}n$. If $U \in \cI_n$, then set $V$ as the closest $v_{n,k}$ to $U$. Otherwise, set $V =0$. Then,
	\begin{equation}
	E[e^{V-U}] \leq e^{2^{-n}}P(U\in \cI_n) + E[e^{-U}\indic\{U \notin \cI_n\}] 
	\end{equation}
	and
	\begin{equation}
	E[V-U] \geq -2^{-n}P(U\in \cI_n) - E[U\indic\{U \notin \cI_n\}].
	\end{equation}
	Since $E[e^{-U}] = 1$, $E[|U|] < \infty$ and $P(U\in \cI_n) \to 1$, the rightmost terms above tend to zero and for large $n$, 
	\begin{equation}
	\log E[e^{V-U}] - E[V-U]
	\end{equation}
is close to zero. The proof will be complete if $I(U;V) < \infty$ for finite $n$. But since $V$ has finite cardinality, this is indeed the case.

	\subsection{Boundedness of $C_U(r)$}\label{app:lipschitz}
	We give the definition of $C_U(r)$ suitable for our setting, adopted from \cite{Kostina}. Recall that the output of the lattice quantization takes values in $\{2kr\}_{k\in\mathbb{Z}}$. Let $I_k := [(2k-1)r,(2k+1)r]$. Then $C_U(r) := \sum_kP(I_k)\max_{I_k}v(u)$.
	
	Now, observe
	\begin{equation}
	E[v(U)] = \sum_k \int_{I_k} p(u)v(u)du.
	\end{equation}
	Since both $p$ and $v$ are continuous, by the mean value theorem, for each $k$ there exists a $c_k$ such that 
	\begin{equation}\int_{I_k} p(u)v(u)du = v(c_k)\int_{I_k} p(u)du.
	\end{equation}
	Suppose $v(u)$ is Lipschitz with constant $L$. Since $a_k := \arg\max_{u\in I_k} v(u)$ has distance at most $2r$ to $c_k$, we have $|a_k-c_k| \leq 2r$ and $v(a_k)\leq 2Lr + v(c_k)$. Then, 
	\begin{equation}
	C_U(r) \leq \sum_k (2Lr + v(c_k))P(I_k) = 2Lr + E[v(U)].
	\end{equation}
	Therefore, finiteness of $E[v(U)]$ guarantees the finiteness of $C_U$. Also observe that if $v(u) = |\frac{d}{du}\log p(u)|$ is Lipschitz, and if $J(U) < \infty$, then $C_U(r)$ is upper bounded as
	\begin{equation}C_U(r) \leq \sqrt{J(U)} + 2Lr\end{equation} 
	and is guaranteed to be finite.
	
	For the case when $v$ is not Lipschitz, but is differentiable and its derivative is Lipschitz with constant $L$, $C_U(r)$ can be upper bounded as
	\begin{equation}\label{eq:higher_lipschitz}
	C_U(r) \leq E[v(U)] + 2rE[|v'(U)|] + 4Lr^2
	\end{equation}
	and if $E[v(U)] $, $E[|v'(U)|]$ are finite, $C_U(r)$ can be bounded from above. Note that similar arguments generalize to higher order derivatives of $v$. 
	
	As an example, suppose $U$ has a density given by  $p(u) = K\exp(-|u|^3)$, with $K$ being the appropriate normalization constant. Then since $\frac{d}{du}\log p(u) = 3u^2$, one cannot find a Lipschitz $v$. Nevertheless, if one sets $v(u) = 3u^2$, $|v'(u)| = 6|u|$ is Lipschitz with constant $L=6$ and since $E[U^2]$, $E[|U|]$ are finite, $C_U(r) $ is bounded according to \eqref{eq:higher_lipschitz}.

\end{document}